\documentclass[UK-english, notap, thm-restate]{lipics-v2021}
\usepackage[utf8]{inputenc}
\usepackage{amsmath}
\usepackage{amsfonts}
\usepackage{amsthm}
\usepackage{makecell}
\usepackage{xspace}
\usepackage{imakeidx}
\usepackage{thmtools}
\usepackage{makecell}
\usepackage{todonotes}
\usepackage{algorithm}
\usepackage{algpseudocode}
\usepackage{booktabs}
\usepackage{ragged2e}
\usepackage{xspace}
\usepackage{adjustbox}
\usepackage{changepage}
\usepackage{todonotes}
\usepackage{tcolorbox}
\usepackage{graphicx,caption}
\usepackage[dvipsnames]{xcolor}
\usepackage{soul}
\soulregister\cite7
\nolinenumbers
\hideLIPIcs

\usepackage[noadjust]{cite}

\newcommand{\eps}{\varepsilon}
\newcommand{\D}{\mathcal{D}}
\newcommand{\T}{\mathcal{T}}
\newcommand{\G}{\mathcal{G}}

\renewcommand{\P}{\mathcal{P}}

\renewcommand{\S}{\mathcal{S}}
\newcommand{\DG}{\mathcal{D}}
\newcommand{\neighb}[1]{#1 * C}
\newcommand{\neighbb}[2]{#1 * #2}
\newcommand{\dist}{d}

\newcommand{\F}{\mathcal{F}}
\newcommand{\M}{\mathcal{M}}
\newcommand{\C}{\mathcal{C}}
\newcommand{\E}{\mathcal{E}}
\newcommand{\Q}{\mathcal{Q}}
\newcommand{\MBM}{\ensuremath{M}}

\newcommand{\typeiiConst}{(2^{i+5}+1)}

\definecolor{lavender}{RGB}{230,230,250}
\NewDocumentCommand{\joe}{sm}{\todo[color=SkyBlue]{Joe: #2}}
\NewDocumentCommand{\joeInline}{sm}{\todo[inline,color=SkyBlue]{Joe: #2}}


\title{
\fontsize{16pt}{20pt}\selectfont A dynamic $(1+\eps)$-spanner for disk intersection graphs}

\titlerunning{ A dynamic $(1+\eps)$-spanner for disk intersection graphs}

\author{Sarita de Berg}{IT University of Copenhagen, Denmark}{debe@itu.dk}{https://orcid.org/0000-0001-5555-966X}{}

\author{Ivor van der Hoog}{IT University of Copenhagen, Denmark}{ivva@itu.dk}{
https://orcid.org/0009-0006-2624-0231}{}

\author{Eva Rotenberg}{IT University of Copenhagen, Denmark}{erot@itu.dk}{
https://orcid.org/0000-0001-5853-7909}{}

\author{Johanne Müller Vistisen}{IT University of Copenhagen, Denmark}{jomy@itu.dk}{https://orcid.org/0009-0007-3008-3061}{}

\author{Sampson Wong}{University of Copenhagen, Denmark}{sampson.wong123@gmail.com}{0000-0003-3803-3804}{}
\authorrunning{S. de Berg, I. van der Hoog, E. Rotenberg, J. Müller Vistisen, and S. Wong}
\Copyright{S. de Berg, I. van der Hoog, E. Rotenberg, J. Müller Vistisen, and S. Wong}


\funding{
This research was supported by Danmarks Frie Forskningsfond Case no. 10.46540/3160-00020B,
the VILLUM Foundation grant (VIL37507) ``Efficient Recomputations for Changeful Problems'',
and 
the European Union's Marie Skłodowska-Curie Actions Postdoctoral Fellowship Project number No. 101146276.}

\begin{document}

\maketitle


\begin{abstract}
 We maintain a $(1+\eps)$-spanner over the disk intersection graph of a dynamic set of disks. We restrict all disks to have their diameter in $[4,\Psi]$ for some fixed and known $\Psi$. The resulting $(1+\eps)$-spanner has size $O(n \eps^{-2} \log \Psi \log (\eps^{-1}))$, where $n$ is the present number of disks.

 We develop a novel use of persistent data structures to dynamically maintain our $(1+\eps)$-spanner.
 Our approach requires $O(\eps^{-2} n \log^4 n \log \Psi)$ space and has an $O( \left( \frac{\Psi}{\eps} \right)^2 \log^4 n \log^2 \Psi \log^2 (\eps^{-1}))$ expected amortised update time.
 For constant $\eps$ and $\Psi$, this spanner has near-linear size, uses near-linear space and has polylogarithmic update time. 
 Furthermore, we observe that for any $\eps < 1$, our spanner also serves as a connectivity data structure. With a slight adaptation of our techniques, this leads to better bounds for dynamically supporting connectivity queries in a disk intersection graph.
 In particular, we improve the space usage when compared to the dynamic data structure of (Baumann et al., DCG'24), replacing the linear dependency on $\Psi$ by a polylogarithmic dependency. Finally, we generalise our results to $d$-dimensional hypercubes.
\end{abstract}

\newpage
\section{Introduction}

Given a set of $n$ geometric shapes, e.g. disks, the corresponding intersection graph has the shapes as its vertex set, and an edge between two shapes whenever they intersect (see Figure~\ref{fig:intro_figure}). Disk intersection graphs are often used as a model for wireless ad-hoc networks~\cite{DBLP:journals/comsur/CadgerCSM13,DBLP:journals/jsac/GaoGHZZ05,DBLP:journals/tac/CortesMB06}, in particular, the radius of a disk equals the transmission range of the device. One of the central algorithmic challenges of disk intersection graphs is that, with $n$ vertices, the graph may have $\Omega(n^2)$ edges. Nonetheless, there are many problems that can be solved in subquadratic time in disk graphs~\cite{KaplanKMRSS2019,klost2023,KaplanKSS23,ConroyT22,clark1990unit,Berg2025,An0025}. 
Both weighted~\cite{FurerK12,An0025,DBLP:journals/siamcomp/GaoZ05,DBLP:journals/jocg/ChanS19a} and unweighted~\cite{DBLP:journals/comgeo/Biniaz20,Berg2025,DBLP:conf/esa/Brewer025,klost2023} disk intersection graphs have been studied; we will consider weighted disk intersection graphs where the weight of an edge is the Euclidean distance between the centres of the disks. 
We study two dynamic problems: maintaining a spanner for the disk intersection graph or a data structure for connectivity queries (for the latter, the edge weights are irrelevant).

\begin{figure}[b]
    \centering
    \includegraphics{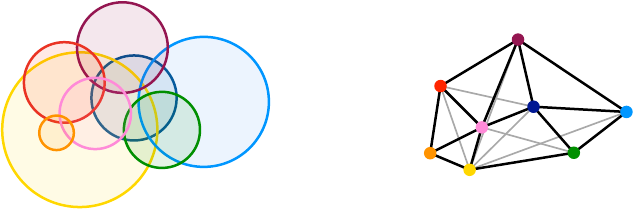}
    \caption{A set of disks~$\S$ (left) and the corresponding disk graph~$\D(\S)$ (right) that includes both the thick black and the thin grey edges. The thick black edges form a $(1+\eps)$-spanner for $\eps = 3/4$.
    }
    \label{fig:intro_figure}
\end{figure}

\subparagraph{Geometric spanners.} Given a graph, a spanner is a subgraph that approximately preserves the distances of the original graph. Formally, a $(1+\varepsilon)$-spanner for a graph $G$ is a subgraph $G'$ such that for every pair of vertices $u,v \in G$, we have that $d_{G'}(u,v) \leq (1+\varepsilon) \cdot d_G(u,v)$, where $d_X(\cdot, \cdot)$ denotes the shortest path distance in~$X$. {Throughout the paper, we assume that $\eps$ is an arbitrary fixed constant with $0< \varepsilon <1$. 
In computational geometry, there are many works that consider spanners of \emph{point sets}, where the graph $G$ is the complete graph over the point set and edge weights correspond to Euclidean distances.
Results for such Euclidean spanners of point sets include a variety of $(1+\eps)$-spanners of size
$O(n/\eps^{d-1})$~\cite{DBLP:journals/dcg/AlthoferDDJS93,DBLP:journals/jacm/CallahanK95,yao1982,DBLP:conf/stoc/AryaDMSS95,DBLP:journals/siamcomp/GudmundssonLN02,DBLP:conf/stoc/Clarkson87,DBLP:journals/siamcomp/LeS25}. For a thorough treatment of spanners for Euclidean point sets, we reference the textbook~\cite{narasimhan2007geometric}. Spanners are also well studied in other settings, such as in doubling spaces~\cite{RoutingDoublingMetrics,DBLP:journals/siamcomp/Har-PeledM06}, hyperbolic spaces~\cite{DBLP:journals/jocg/Kisfaludi-BakW24,dynamic_hyperbolic}, general graphs~\cite{DBLP:journals/dcg/AlthoferDDJS93, peleg1989graph}, minor-free graphs~\cite{DBLP:conf/focs/BorradaileLW17,MinorFree} and polygonal domains~\cite{DBLP:conf/compgeom/AbamAHA15,SpannerPolyhedralTerrain,spanner_planar_domains_via_covers}.
Many of these spanners have also been studied in a dynamic setting, where the goal is to maintain a spanner when vertices are inserted or deleted~\cite{chan2020locality,GottliebR08,dynamic_hyperbolic,ChuzhoyP25,BaswanaKS12}.

\subparagraph{Spanners in disk graphs.} 
There exists a variety of static spanners for unweighted intersection graphs~\cite{Catusse2010PlanarHopSpanners,chan_et_al:LIPIcs.SoCG.2023.23,Dumitrescu2022SparseHopSpanners, ConroyToth2022HopSpanners, Yan2009RoutingUDG}.
Fürer and Kasiviswanathan~\cite{FurerK12} were the first to consider spanners of weighted disk intersection graphs. They constructed a $(1+\eps)$-spanner of size  $O(n\eps^{-2})$ in $O(n^{4/3+\delta}\eps^{-4/3}\log^{2/3}\Psi)$ time, with $\delta \geq 0$, for when all disk diameters lie in $[1, \Psi]$.
 Their results generalise to higher dimensions and other disk-like objects. Kaplan, Mulzer, Roditty, Seiferth, and Sharir~\cite{kaplan2020dynamic} provided an improved $(1+\eps)$-spanner construction that runs in $O(n \eps^{-2}\log^9 n \lambda_6(\log n))$ time, where $\lambda_s(t)$ is the maximum length of a Davenport–Schinzel sequence
on $t$ symbols of order $s$~\cite{DavenportSchinzel}. New techniques by Liu~\cite{liu2022nearestneighbors} improve their construction time to $O(n\eps^{-2}\log^4 n)$ expected time.

 A natural extension is to dynamically maintain a spanner as disks get inserted and deleted. Surprisingly, to the best of our knowledge, this paper is the first to study the dynamic maintenance of a disk intersection graph spanner. This includes weighted, unweighted, and even unit disk intersection graphs.  However, we note that some closely related problems have been studied, such as dynamic connectivity.

\subparagraph{Dynamic connectivity in disk graphs.} In the connectivity problem, the goal is to maintain a data structure on the graph such that connectivity queries can be answered efficiently. A \emph{connectivity query} takes two disks $\sigma,\rho$ as input and returns whether $\sigma$ and $\rho$ are in the same connected component of the disk graph. The most common approach for dynamic connectivity in disk graphs is to construct a proxy graph (a sparser subgraph of the disk intersection graph) and then maintain this sparser graph in the dynamic graph connectivity data structure introduced by Holm, de Lichtenberg, and Thorup~\cite{holm2001poly}.

Kaplan, Mulzer, Roditty, Seiferth, and Sharir \cite{kaplan2020dynamic} study a dynamic set of disks $\S$ whose radii lie in $[1, \Psi]$.
Using the above approach, their data structure supports connectivity queries in $O(\log n / \log \log n)$ and has an expected amortised update time of $O(\Psi^{2}\log^9 n \lambda_6(\log n))$, which can be improved to $O(\Psi^2\log^4 n)$ using \cite{liu2022nearestneighbors}. The overall space is $O(\Psi^2 n \log n)$. Baumann, Kaplan, Klost, Knorr, Mulzer, Roditty, and Seiferth~\cite{baumann2024dynamic} further improve the expected amortised update time to $O(\Psi \log^4 n)$ and the space to $O(\Psi n \log n)$. If the input consists of axis-aligned squares, van der Hoog, Nusser, Rotenberg, and Staals \cite{hoog2024fully} maintain connectivity using  $O(\Psi \log^4 n + \log^6 n)$ amortised update time and $O(n\log^3n \log \Psi)$ space.

Finally, Chan and Huang~\cite{ChanH24a} consider the case where there are no restrictions on the disk radii. They use a completely different approach to obtain a data structure with constant query time and an amortised update time of $O(n^{7/8+\eps})$ for an arbitrarily small $\eps>0$. Polylogarithmic results for the insertion-only and deletion-only cases have been obtained by Baumann, Kaplan, Klost, Knorr, Mulzer, Roditty, and Seiferth~\cite{baumann2024dynamic} using techniques very similar to ours. Nevertheless, this approach has serious robustness issues when studying a fully dynamic set of disks with unbounded radii. One of the reasons for this is that there exist no fully dynamic compressed quadtrees on the real RAM~\cite{har-peled2011geometric,HoogKL18}.

\subparagraph{Our contributions.}  We introduce the first dynamic $(1+\eps)$-spanner for disk intersection graphs. Our model for dynamic disks follows~\cite{kaplan2020dynamic, baumann2024dynamic, hoog2024fully}, whereby each disk has its diameter in $[4, \Psi]$ for a fixed and known $\Psi$. We present a $(1+\eps)$-spanner of $O(n \eps^{-2} \log \Psi \log (\eps^{-1}))$ size. Our expected amortised update time is $O( \left(\frac{\Psi}{\eps}\right)^2 \log^4 n \log^2 \Psi \log^2 (\eps^{-1}))$, and our expected total space is $O(n\eps^{-2} \log^4 n \log \Psi)$.

By using similar techniques, we obtain a dynamic connectivity data structure with $O(\Psi \log^4 n \log \Psi)$ expected amortised update time and $O(n \log^4 n \log \Psi)$ expected space. Compared to the best known result of~\cite{baumann2024dynamic}, we match their update time up to logarithmic factors and improve their space from linear in $\Psi$ to logarithmic in $\Psi$. 

Finally, we note that all our results generalise to $d$-dimensional hypercubes by replacing a disk intersection data structure with a hypercube intersection data structure. Table~\ref{tab:results} provides an overview of our results and the most relevant related work. 

\subparagraph{Organisation.} After covering  the preliminaries in Section~\ref{sec:prelims}, we present an overview of our results and techniques in Section~\ref{sec:overview}. In Section~\ref{sec:high_space}, we present an easier variant of our dynamic $(1+\eps)$-spanner for disks within a fixed bounding box that uses $\tilde{O}(\left(\frac{\Psi}{\eps}\right)^2 n)$ space. In Appendix~\ref{sec:small_space}, we reduce the space usage by a factor $O(\Psi^2)$ and in Appendix~\ref{app:generalisation}, we remove the bounding box assumption. In Appendix~\ref{sec:connectivity}, we present our results on dynamic connectivity. Finally, we generalise all of our results to $d$-dimensional hypercubes in Appendix~\ref{sec:cubes}.

\begin{table}[]
    \centering
    \begin{tabular}{l|l|l|l}
 \multicolumn{4}{c}{\textbf{Dynamic $(1+\eps)$-spanners, fixed aspect ratio $\Psi$}} \vspace{0.1cm}\\
           & Amortised update time & Space & Reference \\
     \hline\hline 
    2D disks  & $O( \left(\frac{\Psi}{\eps}\right)^2 \log^4 n \log^2 \Psi \log^2 (\eps^{-1}))$ exp. & $O( n \, \eps^{-2} \log^4 n \log \Psi)$  exp. & Thm.~\ref{thm:main_small_space} \ \\
    $d$-dim cubes  &  $O( \left(\frac{\Psi}{\eps}\right)^d \log^d n \log^2 \Psi \log^2 (\eps^{-1}))$  &  $O( n \, \eps^{-d} \log^d n \log \Psi)$  & Thm.~\ref{thm:hypercubes_spanner} 
    \vspace{0.4cm}\\ 
   \multicolumn{4}{c}{\textbf{Dynamic connectivity, fixed aspect ratio $\Psi$}} \vspace{0.1cm}\\
    & Amortised update time & Space & Reference \\  \hline\hline
         2D squares  & $O(\Psi \log^4 n+ \log^6 n)$ & $O(n \log^3 n \log \Psi)$ & \cite[Thm. 10]{hoog2024fully}  \\
    2D disks  & $O(\Psi^2 \log^4 n)$ exp. & $O(\Psi^2 n \log n)$ & \cite{kaplan2020dynamic,liu2022nearestneighbors}  \\
         2D disks  &  $O(\Psi \log^4 n)$ exp.  & $O(\Psi n \log n)$ & \cite[Thm.4.6]{baumann2024dynamic} \\   
         2D disks & $O(\Psi \log^4 n \log \Psi)$ exp. & $O(n \log^4 n \log \Psi)$ exp. & Thm~\ref{thm:main_connectivity} \\
         $d$-dim cubes  & $O(\Psi^{d-1}\log^dn\log\Psi)$ & $O(n\log^dn\log\Psi)$  & Thm~\ref{thm:hypercubes_connectivity} \vspace{0.2cm} \\
    \end{tabular}
    \caption{Table of our results. The size of the spanners is $O(n\eps^{-d} \log \Psi \log \eps^{-1})$. All connectivity structures have $O(\log n/\log \log n)$ query time.\vspace{-0.5cm}}
    \label{tab:results}
\end{table}

\section{Preliminaries}\label{sec:prelims} 

The input is a dynamic set $\S = (\sigma_1, \sigma_2, \ldots, \sigma_n)$ of $n$ disks in the plane, or a set of $n$ $d$-dimensional cubes in $\mathbb{R}^d$ for some constant dimension $d$. These preliminaries focus on disks in the plane, but the concepts are defined analogously for $d$-dimensional cubes in $\mathbb{R}^d$.
To minimise numerical dependencies on $\sqrt{2}$, we define for any convex shape $C$ the diameter $|C|$ of $C$ as the length of the largest horizontal segment contained in $C$.
We assume that, at all times, $\forall \sigma \in \S$, the diameter is at least $4$ and at most $\Psi$, where $\Psi$ is some fixed and known value. In other words, $|\sigma| \in [4, \Psi]$. For now, we assume that at all times $\S$ lies contained within a bounding box $B = [0, \Psi^*]^2$, where $\Psi^*$ is the smallest power of two that is bigger than $\Psi$. 
In Appendix~\ref{app:generalisation}, we drop this bounding box assumption.

The \emph{disk intersection graph} $\D(\S)$ is an edge-weighted graph where vertices are the disks in $\S$, and there is an edge $(\sigma_1, \sigma_2)$ whenever the two disks intersect. The weight $\dist(\sigma_1,\sigma_2)$ of $(\sigma_1, \sigma_2)$ is the Euclidean distance between the disk centres. 
For any edge-weighted graph $G$ over $\S$, we let $d_{G}(\sigma_1, \sigma_2)$ denote the length of the shortest path between $\sigma_1$ and $\sigma_2$ in~$G$. 

\subparagraph{Problem statement.}
Since disk intersection graphs can be very dense, one may be interested in maintaining a sparse \emph{spanner} of this graph.
A $(1+\varepsilon)$-spanner of a disk intersection graph $\D(\S)$ is an edge-weighted subgraph $G$ of $\DG(\S)$ 
such that for any $\sigma_1, \sigma_2 \in S$, we have $d_G(\sigma_1, \sigma_2) \leq (1+\varepsilon) \cdot d_{\D(\S)}(\sigma_1, \sigma_2)$. 
The \emph{size} of a spanner is its number of edges. 
We wish to maintain a $(1+\eps)$-spanner of near-linear size subject to disk insertions and deletions in $\S$. 
The running time, space usage, and size of the spanner depend on $n$, $\eps$, and $\Psi$.
Our results apply to both the word RAM, real RAM, and pointer machine model of computation. 

\subparagraph{Dynamic Euclidean spanner.} We will use the following dynamic spanner in our construction. 

\begin{lemma}[\cite{chan2020locality,GottliebR08}]\label{lemma:spanner}
    For any $d$-dimensional point set $P$ of size $n$ there exists a dynamic $(1+\eps)$-spanner of their pairwise Euclidean distances that has size $O(n \eps^{-d} \log (\eps^{-1}))$ that can be maintained using $O(n \eps^{-d} \log (\eps^{-1}))$ space and $O(\eps^{-d} \log n \log^2 (\eps^{-1}))$ update time.
\end{lemma}

\subparagraph{Quadtree.}
We will also use the following quadtree in our construction.

\begin{definition}
\label{definition:complete_quadtree}
Let $B$ be an axis-aligned $d$-dimensional hypercube whose side length is a positive power of two.
A \emph{split} operation selects a hypercube and divides it into $2^d$ equal hypercubes.
A \emph{quadtree} is any tree obtained by recursively applying the split operation on $B$. 
Each node of the tree corresponds to a hypercube in $\mathbb{R}^d$, called a \emph{cell}.
\end{definition}

 For a cell $C$ and any odd integer $k$, we denote by $k * C$ the region $[0,k|C|]^2$ translated to have $C$ as its centre. 
For any disk $\sigma$, we define its \emph{storing cell} $C_\sigma$  (see Figure~\ref{fig:cells}) as the largest cell, obtained by recursively splitting $B$, that: (1) contains the centre of $\sigma$ and (2) is contained in~$\sigma$. 
The \emph{storing family} $\F(\sigma)$ of $\sigma$ is the set of all cells, down to diameter $1$, that can be obtained by recursively splitting $C_\sigma$ and contain the centre of~$\sigma$. 
Typically, one uses the minimum-size quadtree that contains all storing cells $C_\sigma$ for $\sigma \in \S$~\cite{FurerK12, baumann2024dynamic}.  We deviate this standard by considering the minimum-size quadtree that contains all storing families:

\begin{definition}
    \label{definition:quadtree}
    Given a fixed bounding box $B$ and $\S$, we define $T(\S)$ as the minimum-size quadtree that contains the storing families $\F(\sigma)$ for all $\sigma \in \S$.
    For a cell $C \in T(\S)$, the set $\pi(C)$ consists of all disks in $\S$ that have $C$ in their storing families.
\end{definition}

\begin{figure}[]
    \centering
    \includegraphics[page=3]{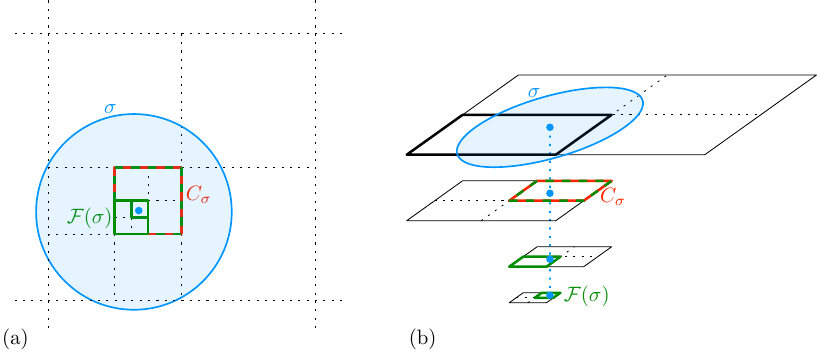}
    \caption{For disk $\sigma$, the storing cell is $C_\sigma$ and its storing family $\F(\sigma)$ in (a) 2D and (b) 3D view.}
    \label{fig:cells}
\end{figure}

In other words, the set $\pi(C)$ of a cell $C$ consists of all disks that have their centre in $C$ and also fully contain $C$. Since $B$ is fixed as $[0, \Psi^*]^2$ we observe: 
\begin{observation}\label{obs_intersection_storing_cell}
The quadtree $T(\S)$ has $O(n \log \Psi)$ cells. Moreover, for all cells $C, C' \in T(\S)$, where $|C| = |C'|$ and $C' \subset 3 * C$, all pairs $(\sigma, \sigma') \in \pi(C) \times \pi(C')$ intersect. 
\end{observation}

Next, we subdivide each cell~$C$ into a finer grid of what we call $\varepsilon$-cells (see Figure~\ref{fig:population}):

\begin{definition}
    \label{definition:eps-squares}
For a fixed $0 < \eps < 1$, we partition each cell $C \in T(\S)$ into a ($d$-dimensional) grid of cells of side length $\eps |C|$, called \emph{$\eps$-cells}. Let $E$ be an $\eps$-cell of a cell $C$, then the \emph{subpopulation} $\Gamma_{\eps}(E)$ denotes all disks that have $C$ as their storing cell and their centre in~$E$. Analogously, the \emph{population} $\pi_{\eps}(E)$ denotes all disks in $\pi(C)$ that have their centre in $E$. 
\end{definition}

\subparagraph{Disk intersection queries.}
Let $\S$ be a set of $n$ disks. Given a query disk $\rho$, the \emph{disk intersection query} structure returns a disk in $\S$ that intersects~$\rho$, if such a disk exists.

\begin{restatable}[\cite{kaplan2020dynamic,liu2022nearestneighbors}]{lemma}{intersectionlemma}
\label{lemma:intersection_queries}
We can maintain a set of $n$ disks $\S \subset \mathbb{R}^2$ in a data structure of $O(n \log n)$ size that can answer intersection queries in $O(\log^2 n)$-time, supporting disk updates (insertions and deletions) in $O(\log^4 n)$ amortised expected time.     
\end{restatable}
\begin{proof}
As observed in~\cite{baumann2024dynamic}, the results of Kaplan, Mulzer, Rodity Seiferth, and Sharir~\cite{kaplan2020dynamic} and Liu~\cite{liu2022nearestneighbors} for dynamic additively weighted nearest neighbour queries can be used to obtain a fully dynamic data structure for disk intersection queries with corresponding space and time bounds as follows. We insert all centres of $\S$ into the data structure, where the weight for each disk is the negative radius of the disk. For an intersection query with a disk $\rho$ with centre point $r$, we query this data structure with $r$. 
Let the point corresponding to a disk $\sigma$ be closest to $r$. If the distance is negative, then $r$ is contained in~$\sigma$.
Otherwise, the returned distance $d$ is the distance from $r$ to the perimeter of~$\S$ and~$\rho$ intersects a disk in $\S$ if and only if $d$ is less than the radius of $\rho$. 
\end{proof}

\begin{figure}[]
    \centering
    \includegraphics[page=1]{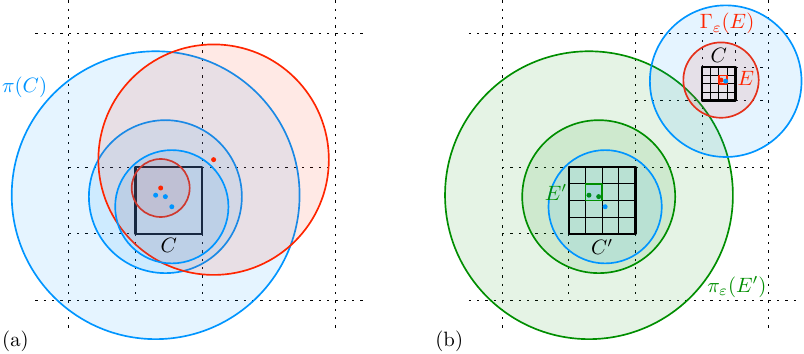}
    \caption{
(a) The set $\pi(C)$ consists of all blue disks, the red disks do not belong to $\pi(C)$ because either their centre is not in $C$, or they do not contain $C$. (b) The $\eps$-cells $E$ and $E'$ of cells $C$ and $C'$ and their corresponding subpopulation $\Gamma_\eps(E)$ (red) and population $\pi_\eps(E')$ (green).
    }
    \label{fig:population}
\end{figure}

\subparagraph{Hypercubes.}
Hypercube intersection queries are defined analogously, and we obtain all our $d$-dimensional hypercube results by using the following lemma instead of Lemma~\ref{lemma:intersection_queries}:

\begin{restatable}[Proof in Appendix~\ref{sec:cubes}]{lemma}{hypercubelemma}
\label{lemma:intersection_queries_cubes}
We can maintain a set of $n$ $d$-dimensional hypercubes $\S \subset \mathbb{R}^d$ in a data structure of $O(n \log^{d-1} n)$ size, that can answer intersection queries in $O(\log^d n)$ time, and that can handle cube updates (insertion or deletion) in $O(\log^d n)$ time.     
\end{restatable}

\subparagraph{Pointer machine data structures.}
In our paper, we will make use of data structure \emph{persistence} as described by Driscoll, Sarnak, Sleator, and Tarjan~\cite{driscoll1986persistent}.
To this end, we require that the data structures that we use (Lemmas~\ref{lemma:intersection_queries} and \ref{lemma:intersection_queries_cubes}) are implementable on a \emph{pointer machine}. We can then either assume a pointer machine as our model of computation, or simulate a pointer machine on the RAM model of our choice.  
A pointer machine consists of \emph{nodes} and \emph{pointers}.
A \emph{data node} stores any type of data (integers, logical values, strings, real numbers, etc), and the dynamic input  (our dynamic set of disks) is a set of data nodes.

A \emph{pointer node} stores a number of constantly many bits.
Each node has constantly many outgoing pointers, where a pointer is a directed edge from one node to another. 
An \emph{update} adds or removes a data node to or from the input. 
A \emph{dynamic data structure} maintains an arbitrarily large set of pointer nodes and, for each data node, a pointer node that points to it. 
One pointer node is called the \emph{root}, and we maintain the invariant that any node can be reached by a directed path from the root. The \emph{size} of the data structure is the total number of nodes it has. 
Any update or query runs a \emph{program} which, starting from the root node, can execute instructions such as \emph{pointer traversals}, \emph{comparisons} between nodes that share a pointer, and \emph{atomic operations} (such as addition between nodes that share a pointer).
The \emph{running time} of a program is the number of instructions it executes before termination. 
The following can be observed from \cite{chan2020dynamic, kaplan2020dynamic,liu2022nearestneighbors} and has been confirmed by~\cite{personal}:

\begin{observation}
\label{obs:pointer}
    Both data structures of Lemmas~\ref{lemma:intersection_queries} and \ref{lemma:intersection_queries_cubes} can be implemented on a pointer machine. However, pointer machine nodes will not necessarily have constant in-degree. 
\end{observation}

\section{Overview of techniques}\label{sec:overview}

The main focus of this paper is to dynamically maintain a $(1+\eps)$-spanner of $\D(\S)$. 
In the main body, we assume that $\S$ remains contained within a bounding box $[0,\Psi^*]^2$, where $\Psi$ is a known upper bound on the diameter of the disks in $\S$ and  $\Psi^* = 2^{\lceil \log_2 \Psi \rceil}$. 
This assumption allows us to ignore dynamic quadtree misalignment issues (see Appendix~\ref{app:generalisation}). The basis of our spanner is a quadtree $T(\S)$ that stores $\S$ (see Definition~\ref{definition:quadtree}).
We then construct a spanner $G$ by using two types of edges: type~\ref{type:one} edges connect two disks in quadtree cells that are ``close'', whereas type~\ref{type:two} edges connect two disks that lie in $\eps$-cells that are ``far away''.  %

We obtain type~\ref{type:one} edges by considering pairs of ``close'' quadtree cells $(C_1, C_2)$, where close is defined such that all disks assigned to $C_1$ intersect all disks assigned to $C_2$. 
We then dynamically maintain a Euclidean $(1+\varepsilon)$-spanner (Lemma~\ref{lemma:spanner}) on the centre points of disks assigned to $C_1$ or $C_2$, and for each edge in this Euclidean spanner, we add a type~\ref{type:one} edge to $G$. 

Type~\ref{type:two} edges are instead constructed for pairs of $\varepsilon$-cells (Definition~\ref{definition:eps-squares}). For any pair $(E, E')$ that lies within a specified distance, we first test whether any two disks in $\Gamma_\eps(E) \times \pi_\eps(E')$ intersect.
One can intuitively imagine this as testing whether any two disks across their populations intersect, but the definition requires slightly more care so as to avoid an additional $O(\log \Psi)$ factor across the paper. 
If there exist any intersecting disks in the bichromatic intersection graph of $\Gamma_\eps(E) \times \pi_\eps(E')$, we add a single edge in the spanner between an arbitrary intersecting pair. 
To dynamically maintain type~\ref{type:two} edges under disk insertions and deletions, we need to be able to maintain whether there is \emph{any} intersection in $\Gamma_\eps(E) \times \pi_\eps(E')$. 
To this end, we maintain a maximal bichromatic matching in this graph. 
To do this, we first show a straightforward approach based on~\cite{kaplan2020dynamic, baumann2024dynamic}:
for any pair $(E, E')$ we maintain two separate disk intersection data structures (Lemma~\ref{lemma:intersection_queries}) storing the disks in $\Gamma_\eps(E)$ and $\pi_\eps(E')$ that do not appear in the matching.
In this way, we can immediately query whether a new disk $\sigma$ can be matched to an existing but unmatched disk. This leads to the following result:

\begin{restatable}{lemma}{mainone}
\label{lemma:main_big_space}
Let $\S$ be a fully dynamic set of disks in the plane with diameters in $[4,\Psi]$ for a fixed and known $\Psi$.
For any fixed $\eps \in (0,1)$, we can maintain a $(1+\eps)$-spanner of $\D(\S)$ of size $O(n \eps^{-2} \log \Psi \log (\eps^{-1}))$ using $O( \left(\frac{\Psi}{\eps}\right)^2 \log^4 n \log \Psi \log^2 (\eps^{-1}))$ expected amortised update time and $O(\left(\frac{\Psi}{\eps}\right)^2  n \log n \log \Psi \log(\eps^{-1}))$ space.
\end{restatable}

The dynamic spanner in Lemma~\ref{lemma:main_big_space} has an $\Theta(\Psi^2)$ dependence in both its update time and space. We observe that the $\Omega(\Psi^2)$ dependence in update time is unavoidable:

\begin{observation}[Figure~\ref{fig:lower_bound}]
\label{observation:lower_bound}
    Any fully dynamic data structure that maintains a $O(1)$-spanner of  $\D(\S)$ 
    has $\Omega( \Psi^2)$ update time. 
\end{observation}

\begin{figure}
    \centering
    \includegraphics{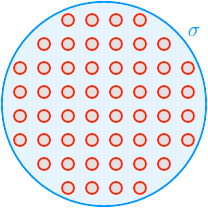}
    \caption{If $\sigma$ is deleted or inserted in $\S$, $\Theta(\Psi^2)$ edges have to be added or removed in the spanner.}
    \label{fig:lower_bound}
\end{figure}

The $\Omega(\Psi^2)$ dependence in space comes from our naive strategy of storing a separate disk intersection data structure for every pair of $\eps$-cells $(E, E')$ at a specified distance. 
To reduce this space dependence, our main insight is to introduce a new form of data structure persistence: \emph{branch persistence} (Definition~\ref{definition:branch_persistence}). This type of persistence allows access and updates to the main (root) version of the data structure and additionally allows access and updates on several \emph{branches}. An update to the root also performs the corresponding update on all branches.
A branch persistent data structure supports branch creation, branch queries, branch-updates and complete rebuilds. Instead of maintaining a disk intersection query data structure (Lemma~\ref{lemma:intersection_queries}) for every pair of suitable $\varepsilon$-cells, which requires $\Omega(\Psi^2)$ space, we instead maintain a branch persistent disk intersection data structure (Lemma~\ref{lemma:intersection_queries_branching}) for each $\eps$-cell. The corresponding disk intersection data structure can then be accessed either via a branch or via the root version of the branch persistent data structure.
By bounding the total symmetric differences between the branches of the new data structure (Lemma~\ref{lemma:bounding_size_of_matching}) and with more careful analysis, we obtain the following dynamic spanner:

\begin{restatable}{thm}{mainTwo}
\label{thm:main_small_space}
Let $\S$ be a fully dynamic set of disks in the plane with diameters in $[4,\Psi]$ for a fixed and known $\Psi$.
For any fixed $0<\eps <1$, we can maintain a $(1+\eps)$-spanner of $\D(\S)$ of size $O(n \, \eps^{-2} \log \Psi \log (\eps^{-1}))$ using $O( \left(\frac{\Psi}{\eps}\right)^2 \log^4 n \log^2 \Psi \log^2 (\eps^{-1}))$ expected amortised update time and $O( n \, \eps^{-2} \log^4 n \log \Psi)$ expected space.
\end{restatable}

\subparagraph{Removing the bounding box assumption.}
The dynamic spanners in Lemma~\ref{lemma:main_big_space} and Theorem~\ref{thm:main_small_space} require that the disks in $\S$ remain in the bounding box $B = [0, \Psi^*]^2$. Given a disk~$\sigma$, our quadtree recursively splits~$B$ until we obtain a unit cell that contains the centre of $\sigma$. 
Unfortunately, if $B$ is unbounded, this procedure takes unbounded time. 

In the static setting, this issue can be avoided by using a compressed quadtree. 
However, despite papers frequently using dynamic compressed quadtrees as an underlying tool, there exists no dynamic compressed quadtree due to the problem of \emph{dynamic quadtree misalignment}. 
As the formal definition of dynamic quadtree misalignment is complex, we defer the definition to Appendix~\ref{app:generalisation}. The issue of dynamic quadtree misalignment was identified in the journal version of~\cite[Appendix A]{baumann2024dynamic}.
There, they describe a complex approach to maintaining connectivity despite quadtree misalignment, which relies on the fact that the disks have a diameter in~$[1, \Psi]$.

In Appendix~\ref{app:generalisation}, we give a simple and general approach that can be applied to~\cite{baumann2024dynamic, hoog2024fully, kaplan2020dynamic}.
Since $\Psi$ is fixed and we only require amortised bounds, it suffices to maintain an insertion-only disjoint collection of boxes $B_i \subset \mathbb{R}^2$ satisfying: (i) each $B_i$ defines an uncompressed quadtree storing $\{ \sigma \in \S \mid B_i \cap \sigma \neq \emptyset \}$, (ii) all $\sigma \in \S$ intersect at least one box $B_i$, and (iii) the union over all uncompressed quadtree spanners is a spanner of $\S$.  
After $N \in \Theta(n)$ updates, we simply rebuild our data structure from scratch --- deleting all $B_i$ that are not intersected by a disk in $\S$. 
As a consequence, we dynamically maintain our disk intersection graph representation with no misalignment issues and no increase in amortised bounds. 

\subparagraph{Connectivity.} 
A $(1+\eps)$-spanner preserves connectivity. By storing our spanner $G$ in the dynamic connectivity data structure of Holm, de Lichtenberg and Thorup~\cite{holm2001poly}, we can also answer connectivity queries in time $O(\log n / \log \log n)$ at no overhead. 
This immediately yields a connectivity data structure that uses less space than the state-of-the-art~\cite{baumann2024dynamic}.
However, its update time depends quadratically on $\Psi$, versus linearly in~\cite{baumann2024dynamic}.
In Appendix~\ref{sec:connectivity}, we combine the techniques in our dynamic spanner with the dynamic connectivity pipeline of~\cite{baumann2024dynamic, hoog2024fully} to obtain a linear dependence on $\Psi$ whilst maintaining the same space usage.
Incorporating our techniques into the pipeline is relatively straightforward, as it simply involves reconstructing their lemmas using our branch persistent data structure. This yields the following theorem:

\begin{restatable}{thm}{mainThree}
\label{thm:main_connectivity}
Let $\S$ be a fully dynamic set of disks in the plane with diameters in $[4,\Psi]$ for a fixed and known $\Psi$.
We can dynamically maintain $\S$ supporting \emph{connectivity queries} in time $O(\log n / \log\log n)$ using $O(\Psi\log^4 n \log\Psi)$ expected amortised update time and $O(n\log^4 n \log \Psi)$ expected space.
\end{restatable}

\subparagraph{Hypercube intersection graphs.}
In Appendix~\ref{sec:cubes}, we show that our techniques extend to intersection graphs of $d$-dimensional hypercubes. We maintain a $(1+\eps)$-spanner and a fully dynamic connectivity data structure for a dynamic set of $n$ $d$-dimensional hypercubes by replacing the disk intersection data structure (Lemma~\ref{lemma:intersection_queries}) with a hypercube intersection data structure (Lemma~\ref{lemma:intersection_queries_cubes}).
For our spanner, the update time depends on a factor $\Theta(\Psi^d)$ as there are $\Theta(\Psi^d)$ unit hypercubes contained in a hypercube of diameter~$\Psi$.
For connectivity, the pipeline from~\cite{baumann2024dynamic, hoog2024fully} shaves a single $\Psi$ factor from the update time (see Theorem~\ref{thm:hypercubes_connectivity}).

\begin{restatable}{thm}{hypercubesSpanner}
\label{thm:hypercubes_spanner}
Let $\S$ be a fully dynamic set of $d$-dimensional axis-aligned hypercubes in $\mathbb{R}^d$  with side lengths in $[4,\Psi]$ for a fixed and known $\Psi$ and constant dimension $d$.
For any fixed $0<\eps <1$, we can maintain a $(1+\eps)$-spanner of $\D(\S)$ of size $O(n \, \eps^{-d} \log \Psi \log (\eps^{-1}))$ using $O( \left(\frac{\Psi}{\eps}\right)^d \log^d n \log^2 \Psi \log^2 (\eps^{-1}))$ amortised update time and $O( n \, \eps^{-d} \log^d n \log \Psi)$ space.
\end{restatable}

\section{Dynamic spanner using $\tilde{O}(\Psi^2 n)$ space within a fixed bounding box}
\label{sec:high_space}

The input is a dynamic set $\S = (\sigma_1, \sigma_2, \ldots, \sigma_n)$ of $n$ disks in the plane, where $\S$ remains contained within the bounding box $[0, \Psi^*]^2$ and  $\forall \sigma \in \S$,  $|\sigma| \in [4, \Psi]$. We denote by $T(\S)$ the dynamic quadtree (Definition~\ref{definition:complete_quadtree}) storing $\S$ which has $O(n \log \Psi)$ space and can trivially be maintained using $O(\log \Psi)$ update time. 
We define a graph $G \subseteq \D(\S)$ that is a $(1+9 \eps)$-spanner of size $O(n \eps^{-2} \log \Psi \log (\eps^{-1}))$. By choosing $\eps' = \eps / 7$, we obtain a $(1+\eps)$-spanner of the same asymptotic size. Our spanner $G$ has two types of edges (see Figure~\ref{fig:the_three_edge_types}):

\begin{figure}[ht]
    \centering
    \includegraphics{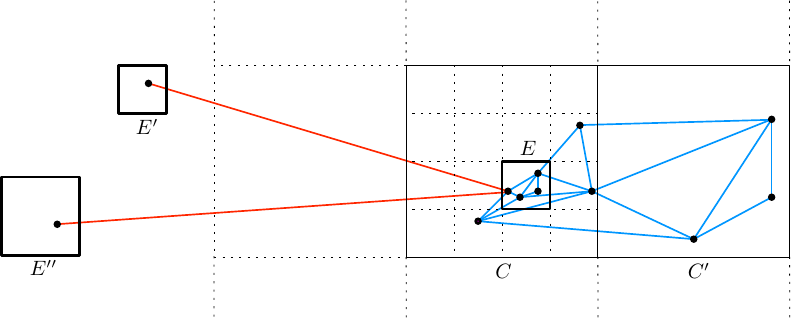}
    \caption{
    The two types of edges where blue edges are type~\ref{type:one} and red edges are type~\ref{type:two}.}
    \label{fig:the_three_edge_types}
\end{figure}

\begin{enumerate}[i.]
    \item \label{type:one} For each cell $C$ in the quadtree $T(\S)$, consider all other cells $C' \in T(\S)$ of size $|C|$ in the region $\neighb{3}$. 
For each pair of cells $(C, C')$, use Lemma~\ref{lemma:spanner} to construct a Euclidean $(1+\eps)$-spanner $\textit{Span}(C,C')$ on the centre points of disks in $\pi(C) \cup \pi(C')$. For every edge in the Euclidean $(1+\eps)$-spanner, add an edge between the pair of corresponding disks to our graph~$G$. By Observation~\ref{obs_intersection_storing_cell}, this edge lies in the disk intersection graph $\D(\S)$.
\item \label{type:two} For each cell $C$ in the quadtree and $i \in [1, \lceil \log \Psi \rceil]$, consider the cells $C'$ such that $|C'| = 2^{i-1} |C|$ and $C'$ intersects $\neighbb{\typeiiConst}{C}$. 
For every pair of $\eps$-cells $E$ of $C$ and $E'$ of $C'$, if there exists a pair $(\sigma, \sigma') \in \Gamma_\eps(E) \times  \pi_\eps(E')$ that intersect, then we select an arbitrary pair of intersecting disks $(\sigma_1, \sigma_2) \in \Gamma_\eps(E) \times \pi_\eps(E')$ and add  $(\sigma_1, \sigma_2)$ to $G$.
\end{enumerate}

\noindent
For type~\ref{type:two} edges, it will be useful to consider $(\sigma_1, \sigma_2)$ to be an ordered pair. In other words, if we write that $(\sigma_1, \sigma_2)$ is a type~\ref{type:two} edge, then we have that $\sigma_1 \in \Gamma_\eps(E)$ and $\sigma_2 \in \pi_\eps(E')$, where $E$ is an $\varepsilon$-cell of $C$, $E'$ is an $\varepsilon$-cell of $C'$, $|C'| = 2^{i-1} |C|$ and $C'$ intersects $(2^{i+5} + 1) * C$.

We show that $G$ is a $(1+9 \eps)$-spanner of $\DG(\S)$. 
Specifically, we will show that for every edge $(\sigma_1,\sigma_2)\in \DG(\S)$ there exists a path in $G$ from $\sigma_1$ to $\sigma_2$ of length at most $(1+9 \eps)\cdot d(\sigma_1,\sigma_2)$.
To this end, consider the storing cells $C_1$~of~$\sigma_1$, and $C_2$~of~$\sigma_2$. Assume without loss of generality that $|C_1| \leq |C_2|$. We have two cases: either the centre of $\sigma_2$ lies inside $\neighbb{3}{C_1}$ or not.

\begin{lemma}\label{lem:spanner_short_edge}
If the centre of $\sigma_2$ lies in $\neighbb{3}{C_1}$ then there exists a path in $G$ from $\sigma_1$ to $\sigma_2$ of length at most $(1+\eps) \cdot d(\sigma_1,\sigma_2)$.
\end{lemma}

\begin{proof}
    Since the centre of $\sigma_2$ lies in  $\neighbb{3}{C_1}$, the cell $C_2$ must have a descendent $C_2'$ of size $|C_1|$ that is contained in $\neighbb{3}{C_1}$ where $\sigma_2 \in \pi(C_2')$.
    By construction of the type~\ref{type:one} edges, $G$ contains an Euclidean $(1+\eps)$-spanner over the centre points of all the disks in $\pi(C_1) \cup \pi(C_2')$, so there is a path of length $\leq (1+\varepsilon) \cdot d(\sigma_1, \sigma_2)$ from $\sigma_1$ to $\sigma_2$ using only type~\ref{type:one} edges.
\end{proof}

\begin{lemma}\label{lem:spanner_long_edge}
   If the centre of $\sigma_2$ does not lie in $\neighbb{3}{C_1}$ then there exists a path in $G$ from $\sigma_1$ to $\sigma_2$ of length at most $(1+9 \eps) \cdot \dist(\sigma_1, \sigma_2)$. 
\end{lemma}
\begin{proof}
Intuitively, since the centre of $\sigma_2$ is ``far away'' relative to $\sigma_1$, we will require both type~\ref{type:one} and type~\ref{type:two} edges to construct a spanning path from $\sigma_1$ to $\sigma_2$.
Recall that $C_1$ is the storing cell of $\sigma_1$ and consider the storing family $\F(\sigma_2)$ of $\sigma_2$.
The first step is to find a cell $C'$ in the storing family~$\F(\sigma_2)$ such that there exists a type~\ref{type:two} edge between $C_1$ and $C'$. We will consider two cases, depending on the relative sizes of $|C_1|$ and $\dist(\sigma_1, \sigma_2)$. Let $\Delta = \dist(\sigma_1, \sigma_2)$. 
Since the centre of $\sigma_2$ is not in $\neighbb{3}{C_1}$, we have $\Delta \geq |C_1|$, so there exists an integer $j \geq 0$ such that $\Delta \in [2^j |C_1|, 2^{j+1}|C_1|)$. 
Our two cases are $0\leq j \leq 3$ and $ j > 3$. 

\begin{itemize}

\item If $0\leq j \leq 3$ then $\Delta \in [|C_1|, 2^{4}|C_1|)$. As $|C_2| \geq |C_1|$, there exists a cell $C'\in \F(\sigma_2)$ for which $|C'| =|C_1|$. We also know that $C'$ lies in $\neighbb{(2^5+1)}{C_1}$ because $\Delta < 2^{4}|C_1|$. Therefore, the definition of type~\ref{type:two} edges applies to the cells $C_1$ and $C'$ (by setting $i=1$ in the definition), and there exists a type~\ref{type:two} edge between $C_1$ and $C'$. 

\item Now for the case of $j>3$, we consider the positive integer $i := j-3$, i.e. $\Delta \in [2^{i+3} |C_1|, 2^{i+4}|C_1|)$. The first observation is then that $\Delta$ is at most $\frac{1}{2}(|\sigma_1| + |\sigma_2|)$. 
Moreover, since $|C_2| \geq |C_1|$, $|\sigma_2| \geq \frac{1}{4} |\sigma_1|$ and so $\Delta \leq \frac{5}{2} |\sigma_2|$. Because $|C_2| \geq \frac{|\sigma_2|}{4\sqrt{2}}$ (see Figure~\ref{fig:ratiodiffcellanddisk}) it follows that $|C_2|\geq \frac{2}{20\sqrt{2}} \Delta \geq \frac{2}{20\sqrt{2}}2^{i+3}|C_1| \geq 2^{i-1}|C_1|$. This means that there exists a $C' \in \F(\sigma_2)$ with $|C'| = 2^{i-1}|C_1|$. Because $\sigma_1$ and $\sigma_2$ intersect and their centres are contained in~$C_1$ and~$C'$, the distance between $C_1$ and $C'$ is at most $\Delta$, which is upper bounded by  $2^{i+4} |C_1|$. Thus, we have that $C'$ intersects $\neighbb{(2^{i+5}+1)}{C_1}$ and $|C'| = 2^{i-1}|C_1|$, and therefore the definition of type~\ref{type:two} edges applies to the cells $C_1$ and~$C'$.
\end{itemize}

In both cases, we have found a cell $C'$ such that the definition of type~\ref{type:two} edges applies to $C_1$ and $C'$. We notice that, when applying the definition of type~\ref{type:two} edges, it provides us with a  type~\ref{type:two} edge between every pair of $\varepsilon$-cells that contain a pair of intersecting disks. 
Specifically, let $E_1$ be the $\varepsilon$-cell of $C_1$ that contains the centre of $\sigma_1$, and let $E_2$ be the $\varepsilon$-cell of $C'$ that contains the centre of $\sigma_2$. The above analysis provides a type~\ref{type:two} edge between an arbitrary pair of intersecting disks $(\sigma_1^*, \sigma_2^*)\in \Gamma_\eps(E_1) \times  \pi_\eps(E_2)$ that is in $G$. Per construction, the centres of $\sigma_1, \sigma_1^*$ lie in $E_1$, and the centres of $\sigma_2, \sigma_2^*$ lie in $E_2$, see Figure~\ref{fig:approximate_spanner_path}.

\begin{figure}
    \centering
    \includegraphics{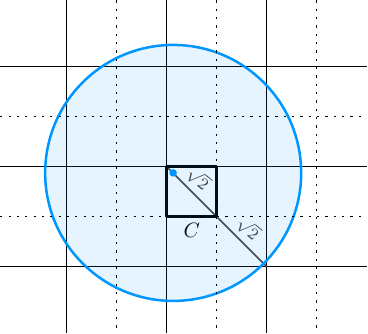}
    \caption{A disk has diameter at most $4\sqrt{2}$ times the diameter of its storing cell.}
    \label{fig:ratiodiffcellanddisk}
\end{figure}

We upper bound the length of the path $\sigma_1 \rightsquigarrow \sigma_1^* \to \sigma_2^* \rightsquigarrow \sigma_2$. Since $\sigma_1$ and $\sigma_1^*$ are both in $E_1$, we can apply Lemma~\ref{lem:spanner_short_edge} to observe that there exists a path in $G$ from $\sigma_1$ to $\sigma_1^*$ of length at most $(1+\eps) d(\sigma_1, \sigma_1^*) \leq (1+\eps) \sqrt{2} \eps |C_1|\leq 2 \sqrt{2} \eps\Delta \leq 3 \eps \Delta$. Similarly, there exists a path in~$G$ from $\sigma_2$ to $\sigma_2^*$ of length at most $(1+\eps) d(\sigma_2, \sigma_2^*) \leq (1+\eps) \sqrt{2}\eps |C'| \leq 3 \eps \Delta$. 
Since $d(\sigma_1,\sigma_2) = \Delta$,  $d(\sigma_1, \sigma_1^*) \leq \sqrt{2}\eps \Delta$, and $d(\sigma_2, \sigma_2^*) \leq \sqrt{2}\eps \Delta$. Finally, by the triangle inequality, the length of the type~\ref{type:two} edge $(\sigma_1^*, \sigma_2^*)$ is at most $(1+3\varepsilon)\Delta$.
Therefore, the length of a spanning path from $\sigma_1 \rightsquigarrow \sigma_1^* \to \sigma_2^* \rightsquigarrow \sigma_2$ is at most $3 \eps \Delta + (1+3\varepsilon)\Delta + 3 \eps \Delta \leq (1+9 \varepsilon) \Delta$.
\end{proof}

\noindent
Next, we bound the size of the spanner.
\begin{figure}
    \centering
    \includegraphics{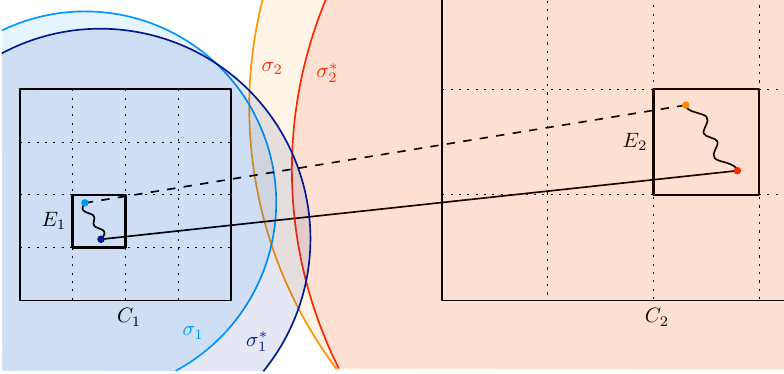}
    \caption{The path $\sigma_1 \rightsquigarrow \sigma_1^* \rightarrow \sigma_2^* \rightsquigarrow \sigma_2$ from the proof of Lemma~\ref{lem:spanner_long_edge}.
    }
    \label{fig:approximate_spanner_path}
\end{figure}

\begin{lemma}\label{thm:G_is_a_spanner}
\label{thm:spanner-size}
The graph $G$ is a $(1+\eps)$-spanner of size $O(n \eps^{-2} \log \Psi \log (\eps^{-1}))$.
\end{lemma}
\begin{proof}
Lemmas~\ref{lem:spanner_short_edge} and \ref{lem:spanner_long_edge} imply that $G$ is a $(1+9 \eps)$-spanner. 
By computing a $(1+\eps')$ spanner with $\eps' = \frac{\eps}{9}$, we thus obtain a $(1+\eps)$-spanner at no asymptotic size overhead.
It remains to bound the size of $G$, i.e. the number of type~\ref{type:one} and type~\ref{type:two} edges in~$G$.
We first bound the number of type~\ref{type:one} edges. Every cell $C \in T(\S)$ has at most $O(1)$ cells $C'$ of size $|C|$ such that $C' \in \neighbb{3}{C}$ or $C \in \neighbb{3}{C'}$.  
For each pair $(C,C')$, by Lemma~\ref{lemma:spanner}, we add $O( \left(|\pi(C)| + |\pi(C')| \right) \eps^{-2} \log (\eps^{-1})  )$ type~\ref{type:one} edges to the subgraph $\pi(C) \cup \pi (C')$. If we uniformly charge the edges in the subgraph, then each disk $\sigma \in \pi(C) \cup \pi (C')$ is charged $O(\eps^{-2} \log (\eps^{-1}))$ times. Each $\sigma$ can participate in $O(\log \Psi)$ cells~$C$ in the quadtree, so each disk~$\sigma$ can be charged at most $O(\eps^{-2} \log \Psi \log (\eps^{-1}))$ times. Therefore, the total number of type~\ref{type:one} edges is $O(n \eps^{-2} \log \Psi \log (\eps^{-1}))$.

We also bound for each disk $\sigma \in \S$ the number of  type~\ref{type:two} edges $(\sigma, \sigma')$ that $\sigma$ participates in, where we recall that type~\ref{type:two} edges $(\sigma, \sigma')$ are an ordered pair.
For each disk $\sigma$ with storing cell $C$, there is a unique $\eps$-cell $E$ of $C$ that contains the centre point of $\sigma$. 
For each index $i \in [1,\lceil \log \Psi \rceil]$, there are constantly many cells $C'$ of size $2^{i-1}|C|$ that intersect $\neighbb{\typeiiConst}{C}$.
For every such cell $C'$, there are $O(\eps^{-2})$ $\eps$-cells $E'$ of $C'$. 
For $(E, E')$ there exists at most one type~\ref{type:two} edge $(\sigma_1, \sigma_2) \in \Gamma_\eps(E) \times \pi_\eps(E')$ (where $\sigma_1$ \emph{may} be $\sigma$).
Therefore, there can be at most $O(\eps^{-2} \log \Psi)$ edges $(\sigma, \sigma')$ of type~\ref{type:two} for fixed $\sigma$.
\end{proof}

\begin{lemma}
    Given a set of disks $\S$ contained in the bounding box $[0, \Psi^*]^2$, we can construct the spanner $G$ in $O(n \eps^{-2} \log^4 n \log \Psi \log^2 (\eps^{-1}))$~time. 
\end{lemma}

\begin{proof}
We can construct the quadtree and compute for each cell $C$ the set $\pi(C)$ in $O(n \log \Psi)$ total time.
There are at most $O( n \log \Psi)$ $\eps$-cells $E$ with non-empty population $\pi_\eps(E)$.
In $O(n \log \Psi)$ time, we can compute the population $\pi_\eps(E)$ and the subpopulation $\Gamma_\eps(E)$ for all of these cells. 
Next, we consider all $O(n \log \Psi)$ cells $C$ for which $\pi(C)$ is non-empty and the $O(1)$ cells $C'$ for which there are type~\ref{type:one} edges between disks in $C$ and $C'$. 
Constructing the Euclidean spanner between $C$ and $C'$ takes $O( \left( |\pi(C)| + |\pi(C')| \right) \eps^{-2} \log n \log^2 (\eps^{-1}))$ time. Since each $\sigma \in S$ appears in $O(\log \Psi)$ sets $\pi(C^*)$ for $C^* \in T(\S)$, doing this for all such pairs $(C, C')$ takes $O(n \eps^{-2} \log n \log \Psi \log^2 (\eps^{-1})) $ total time. 

Next, we consider all $\eps$-cells $E$ for which the population $\pi_\eps(E)$ is non-empty, and we store $\pi_\eps(E)$ in the data structure of Lemma~\ref{lemma:intersection_queries}, which takes $O(n \log^4 n \log \Psi)$ expected total time. Let $E$ be an $\eps$-cell of a quadtree cell $C$. We consider for all $i \in [1, \lceil \log \Psi \rceil]$ all $O(1)$ cells $C'$ of size $2^{i - 1}|C|$ that intersect $\neighbb{\typeiiConst}{C}$.
For each of the $O(\eps^{-2})$ $\eps$-cells $E'$ of $C'$, we query for each disk $\sigma \in \Gamma_\eps(E)$ whether it intersects a disk $\sigma'$ in $\pi_\eps(E')$ in $O(\log^2 n)$ time.
If so, then we add $(\sigma, \sigma')$ to $G$ and continue to some other $\eps$-cell $E''$.
Since each disk appears in exactly one subpopulation $\Gamma_\eps(E)$ this takes $O(n \eps^{-2} \log^2 n \log \Psi)$ total time. 
\end{proof}

\subsection{Dynamically maintaining our spanner }\label{sec:dynamically_maintaining_the_spanner}

We show how to maintain~$G$ dynamically under disk insertions and deletions. We will maintain the type~\ref{type:one} and type~\ref{type:two} edges independently. Maintaining the type~\ref{type:one} edges will be relatively straightforward, as it essentially only requires us to invoke the dynamic updates from Lemma~\ref{lemma:spanner}. Maintaining the type~\ref{type:two} edges will require more work.

Following the definition of type~\ref{type:two} edges, let us consider pairs of $\varepsilon$-cells $(E,E')$ defined as follows: for every cell~$C$ in the quadtree $T(\S)$, for every $i \in [1,\lceil \log \Psi \rceil]$, and for every cell~$C'$ where $|C'| = 2^{i-1} |C|$ and $C'$ intersects $(2^{i+5}+1) * C$, consider pairs $(E,E')$ where $E$ is an $\varepsilon$-cell of $C$ and $E'$ is an $\varepsilon$-cell of $C'$. Recall that there is a type~\ref{type:two} edge between $E$ and $E'$ if and only if there exists a pair of disks $(\sigma, \sigma') \in \Gamma_\eps(E) \times \pi_\eps(E')$ that intersect. Our approach will be to maintain a maximal matching in the bichromatic disk intersection graph of disks in  $\Gamma_\eps(E) \times \pi_\eps(E')$. To this end, we will construct, separately for every pair of $\varepsilon$-cells $(E,E')$, the following three auxiliary data structures (see Figure~\ref{fig:auxiliary_ds}).

\begin{definition}\label{def:auxiliary_ds} For every pair of $\eps$-cells $(E,E')$, where for their quadtree cells $C$ and $C'$ it holds that $|C'| = 2^{i-1} |C|$ for some $i \in [1,\lceil \log \Psi \rceil]$ and $C'$ intersects $(2^{i+5}+1) * C$, we:
\begin{itemize}
    \item Store a maximal bichromatic matching $\MBM(E,E')$ of the bipartite graph $\Gamma_\eps(E) \times \pi_\eps(E')$, where there is an edge $(\sigma, \sigma') \in \Gamma_\eps(E) \times \pi_\eps(E')$ if and only if $\sigma$ and $\sigma'$ intersect. 
    \item Store a disk intersection query structure $Q(E,E')$, defined in Lemma~\ref{lemma:intersection_queries}, for all disks in $\Gamma_\eps(E)$ that do not appear in the matching $\MBM(E,E')$. 
    \item Store a disk intersection query structure $Q'(E,E')$, defined in Lemma~\ref{lemma:intersection_queries}, for all disks in $\pi_\eps(E')$ that do not appear in the matching $\MBM(E,E')$. 
\end{itemize}
\end{definition}

\begin{figure}
    \centering
    \includegraphics{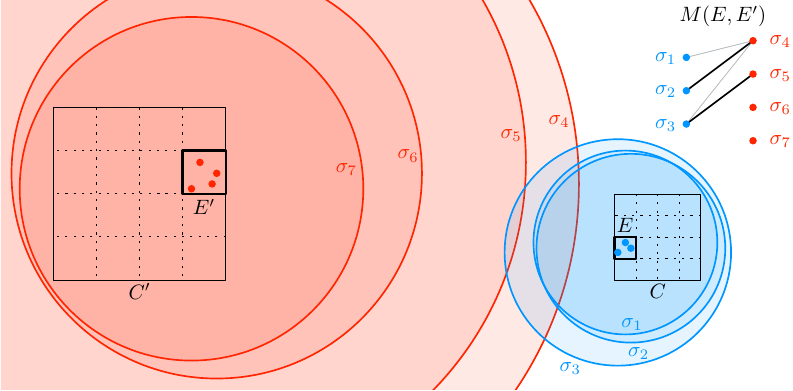}
    \caption{The bichromatic disk intersection graph of disks in $\Gamma_\eps(E)$ (blue) and $\pi_\eps(E')$ (red) and the maximal matching $M(E,E')$ (black). $Q(E,E')$ contains $\sigma_1$, and $Q'(E,E')$ contains $\sigma_6$ and $\sigma_7$.}
    \label{fig:auxiliary_ds}
\end{figure}

Now that we have set up the three auxiliary structures, we are ready to prove the main theorem of the section, which is to dynamically maintain the spanner~$G$.

\mainone*

\begin{proof}
We assume that $\S$ remains inside $[0, \Psi^*]^2$ (we remove this assumption in Appendix~\ref{app:generalisation}). 
As we maintain the same spanner as in Lemma~\ref{thm:spanner-size}, the spanner size bound follows.

\subparagraph{Notation.} Consider the quadtree cells of at least unit size obtained by recursively splitting the bounding box. 
We let $\mathcal C$ be the set of pairs of such cells $(C,C')$ satisfying: $|C| = |C'|$ and $C' \subset 3*C$. Recall that $\textit{Span}(C,C')$ is a dynamic Euclidean spanner on the points $\pi(C) \cup \pi(C')$. Let $\mathcal E$ be the set of pairs of $\varepsilon$-cells $(E,E')$ satisfying: $E$ is an $\varepsilon$-cell of $C$, $E'$ is an $\varepsilon$-cell of $C'$, and there is an $i \in [1,\lceil \log \Psi \rceil]$ such that $|C'| = 2^{i-1} |C|$ and $C' \subset (2^{i+5}+1) * C$. 

\subparagraph{Disk insertions.} When inserting a disk $\sigma$ into $\S$ we perform the following operations:
\begin{enumerate}
    \item We consider all quadtree cells, down to unit size, that are contained in $\sigma$ and contain the centre of $\sigma$. This is the storing family $\F(\sigma)$ of $\sigma$. For each cell $C \in \F(\sigma)$ we either identify $C$ in $T(\S)$ or add $C$ to $T(\S)$. 
    \item We consider all quadtree cells $C \in \F(\sigma)$. For every quadtree cell $C'$ such that $(C,C') \in \mathcal C$, we insert the centre of $\sigma$ into the Euclidean spanner $\mathit{Span}(C,C')$. 
    \item We consider the unique storing cell $C_\sigma$ of $\sigma$ and the unique $\eps$-cell $E$ of $C_\sigma$ that contains the centre of $\sigma$. 
    We consider all $E'$ such that $(E,E') \in \mathcal E$ and test whether $\sigma$ intersects a disk~$\sigma'$ in $\pi_\eps(E')$ that is not already in the matching  $\MBM(E,E')$ by querying $Q'(E,E')$.
    \begin{itemize}
        \item If so, we add $(\sigma, \sigma')$ to $\MBM(E,E')$ and delete $\sigma'$ from $Q'(E,E')$. \item Otherwise, we add $\sigma$ to $Q(E,E')$.
    \end{itemize}
    \item To better match Definition~\ref{def:auxiliary_ds}, we now add a prime to our notation and consider all $C' \in \F(\sigma)$. There is a unique $\eps$-cell $E'$ of $C'$ that contains the centre of $\sigma$. We consider all $\eps$-cells $E$ such that $(E,E') \in \mathcal E$.
    We test whether $\sigma$ intersects a disk~$\sigma''$ in  $\Gamma_\eps(E)$ that is not already in the maximal bichromatic matching  $\MBM(E,E')$ by querying $Q(E,E')$.
    \begin{itemize}
        \item If so, we add $(\sigma'', \sigma)$ to $\MBM(E,E')$ and delete $\sigma''$ from $Q(E,E')$.  
        \item Otherwise, we add $\sigma$ to $Q'(E,E')$.
    \end{itemize}
    \item For every edge updated in any Euclidean spanner $\mathit{Span}(C,C')$, we update the corresponding type~\ref{type:one} edge in $G$ and whenever an empty bichromatic matching~$\MBM(E,E')$ becomes non-empty, we also add the corresponding type~\ref{type:two} edge to~$G$.
\end{enumerate}

\subparagraph{Disk deletions.}
When deleting a disk $\sigma$ from $\S$ we perform the following operations:
\begin{enumerate}
        \item We consider all quadtree cells, down to unit size, that are contained in $\sigma$ and contain the centre of $\sigma$. This is the storing family $\F(\sigma)$ of $\sigma$. If every cell $C \in T(\S)$ records $|\pi(C)|$, then we can simply iterate over all cells $C \in \F(\sigma)$ and decrement this counter.
        At the end of the procedure, we then remove all cells $C \in F(\sigma)$ for which we set the counter to zero, to update $T(\S)$. 
    \item For all cells $C \in \F(\sigma)$, for every quadtree cell $C'$ such that $(C,C') \in \mathcal C$, we delete the centre of $\sigma$ from the Euclidean spanner $\mathit{Span}(C,C')$. 
    \item We consider the unique storing cell $C_\sigma$ of $\sigma$ and its $\eps$-cell $E$ that contains its centre.  We consider all $E'$ such that $(E,E') \in \mathcal E$ and query whether $\sigma$ is in the matching $M(E,E')$.
    \begin{itemize}
        \item If there is an edge $(\sigma,\sigma') \in M(E,E')$,  we delete $(\sigma,\sigma')$ from the matching. To maintain a maximal matching, we need to test whether we can re-match $\sigma'$. We query whether $\sigma'$ intersects a disk~$\sigma''$ in $\Gamma_\eps(E)$ that is not already in $\MBM(E,E')$ by querying $Q(E,E')$.
        \begin{itemize}
            \item If so, we add $(\sigma'', \sigma')$ to $M(E,E')$ and delete $\sigma''$ from $Q(E,E')$. 
            \item Otherwise, we insert $\sigma'$ into $Q'(E,E')$.
        \end{itemize}
        \item Otherwise, if there is no edge $(\sigma,\sigma') \in M(E,E')$ then we must delete $\sigma$ from $Q(E,E')$.
    \end{itemize}
    \item To better match Definition~\ref{def:auxiliary_ds}, we now add a prime to our notation and consider all $C' \in \F(\sigma)$. There is a unique $\eps$-cell $E'$ of $C'$ which contains the centre of $\sigma$. We consider all $\eps$-cells $E$ such that $(E,E') \in \mathcal E$ and query whether $\sigma$ is in the matching $M(E,E')$.
    \begin{itemize}
        \item If there is an edge $(\sigma'', \sigma)$ in $M(E,E')$ then we delete $(\sigma'', \sigma)$ from the matching. To maintain a maximal matching, we consider re-matching $\sigma''$ by querying whether $\sigma''$ intersects a disk~$\sigma' \in \pi_\eps(E')$ that is not already in  $\MBM(E,E')$, by querying $Q'(E,E')$.
        \begin{itemize}
            \item If so, then we add $(\sigma'', \sigma')$ to $M(E,E')$ and delete $\sigma'$ from $Q'(E,E')$. 
            \item Otherwise, we insert $\sigma''$ into $Q(E,E')$.
        \end{itemize}
        \item If $\sigma$ is not matched in $\MBM(E,E')$ then we delete $\sigma$ from $Q'(E,E')$.
    \end{itemize}
    \item For every updated Euclidean spanner $\mathit{Span}(C,C')$, we update the corresponding type~\ref{type:one} edge in $G$ and whenever a matching~$\MBM(E,E')$ becomes empty (or, when its corresponding edge in $G$ had $\sigma$ as an endpoint) we update the corresponding type~\ref{type:two} edge in~$G$. 
\end{enumerate}\vspace{-0.3cm}
\subparagraph{Update time.} Let $\sigma$ be either an inserted or a deleted disk. 
There are $O(\log \Psi)$ quadtree cells $C \in \F(\sigma)$. For each $C$ there are $O(1)$ quadtree cells $C'$ such that $(C,C') \in \mathcal C$ for which we maintain a Euclidean spanner $\textit{Span}(C,C')$. Inserting or deleting a point from this spanner takes $O(\eps^{-2} \log n \log^2 (\eps^{-1}))$ time. In total, updating the Euclidean spanners and the type~\ref{type:one} edges of $G$ (step 2 and 5) takes $O(\eps^{-2} \log n \log \Psi \log^2 (\eps^{-1}))$ time per update in $\S$.

We note that steps 3 and 4 execute the same procedure, but the more expensive operation is step 4, which first iterates over all $O(\log \Psi)$ cells $C' \in \F(\sigma)$.
Every such cell $C'$ has a unique $\varepsilon$-cell $E'$ containing the centre of $\sigma$, and for each fixed $E'$, there are $O(\Psi^2 \varepsilon^{-2})$ $\varepsilon$-cells $E'$ such that $(E,E') \in \mathcal E$. 
Indeed, the definition of $\mathcal{E}$ involves a pair of quadtree cells $(C, C')$ where $C$ is smaller than $C'$. 
For a fixed quadtree cell, $C'$, for the bottom level of the quadtree (the level where all quadtree cells have unit size) there are $O(\Psi^2)$ cells that can lie within distance $O(\Psi)$ from $C$ (every such cell $C$  generates at most $O(\eps^{-2})$ pairs  $(E,E') \in \mathcal E$).
For the level above, there are half as many, and so by a geometric series we obtain that in Step 4 there are $O(\Psi^2 \varepsilon^{-2})$ pairs $(E, E')$. For each pair $(E,E')$ we perform $O(1)$ queries and updates on the three auxiliary data structures. The running time for these queries and updates is dominated by the $O(\log^4 n)$ amortised expected time to update $Q(E,E')$ and $Q'(E,E')$ (Lemma~\ref{lemma:intersection_queries}). In total, updating $M(E,E')$, $Q(E,E')$, $Q'(E,E')$ and the type~\ref{type:two} edges of $G$ (step 3, 4, and 5) takes $O(\Psi^2 \varepsilon^{-2} \log^4 n \log \Psi)$ expected amortised time per update.

\subparagraph{Space usage.} We analyse the space usage of the three main components (the Euclidean spanner, the maximal bichromatic matchings, and the intersection data structures) separately.

The size and space usage of the Euclidean spanner are asymptotically the same. As before, the total space of all Euclidean spanners is thus $O(n \eps^{-2} \log \Psi \log (\eps^{-1}))$.
Next, we analyse the total size of all bichromatic matchings. Consider a cell $C$ in the quadtree and an $\eps$-cell $E$ of $C$. There are $O(\eps^{-2} \log \Psi)$ cells $E'$ such that $(E,E') \in \E$. We thus consider only $O(\eps^{-2}\log \Psi)$ maximal bichromatic matchings for $E$. As each disk appears in the subpopulation $\Gamma_\eps(E)$ of exactly one $\eps$-cell $E$, the total number of edges in all matchings is $O(n\eps^{-2}\log \Psi)$.

Finally, we analyse the size of the disk intersection data structures. A disk appears in disk intersection structures $Q'(E,E')$ for $O(\log \Psi)$ $\eps$-cells $E'$, at most once per quadtree level. For such an $\eps$-cell $E'$, there are $O(\Psi^2 \eps^{-2})$ $\eps$-cells $E$ for which we store a $Q'(E,E')$ data structure. Note again that by a geometric series over all quadtree levels, it follows that for a fixed $\eps$-cell $E'$ there are a total of $O(\Psi^2 \eps^{-2})$ $\eps$-cells $E$ for which we store a $Q'(E,E')$ data structure.
As the size of the disk intersection data structure is $O(n \log n)$, the total size of all of the $Q'(E,E')$ data structures is $O(n \Psi^2 \eps^{-2} \log n \log \Psi)$. The total space of the $Q(E,E')$ data structures is dominated by the $Q'(E,E')$ data structures. 
\end{proof}

\bibliographystyle{plainurl}
\bibliography{references}
\appendix
\newpage




\section{Achieving $\tilde{O}(n)$ space within a fixed bounding box}\label{sec:small_space}

The data structure of Lemma~\ref{lemma:main_big_space} has a quadratic size-dependency on the maximum disk size $\Psi$.
This is primarily due to the fact that 
the data structure requires for every $\eps$-cell that stores a set of disks $\S'$,  $O(\Psi^2\eps^{-2})$ copies of an intersection data structure, each of which includes some subset of $\S'$.
We argue that in many of these cases, these copies are near-identical. 
We show that we can obtain better space-bounds by making the data structure sensitive to the number of edges across all maximal bichromatic matchings instead.

We will achieve this through a form of data structure persistence. 
Persistence is traditionally a way to retain older versions of a dynamic data structure. \textit{Partial persistence} allows queries to all previous versions of the data structure, but only allows updates to the latest version. \textit{Full persistence} allows both queries and updates to any version of the data structure. There are various ways to make a data structure persistent, as described by Driscoll, Sarnak, Sleator, and Tarjan~\cite{driscoll1986persistent}.
The type of persistence we require slightly differs from both partial and full persistence: it is more general than partial persistence, but because our auxiliary intersection data structures cannot natively be implemented on a pointer machine whilst maintaining constant in-degree, we cannot easily use full persistence.
We therefore define an alternative type of persistence, called \emph{branch persistence}. To be precise, our goal is to make the intersection data structure of Lemma~\ref{lemma:intersection_queries} branch persistent. 

We adapt the \emph{fat node persistence} technique from~\cite{driscoll1986persistent}, which is applicable to all data structures that work in the pointer machine model (in \cite{driscoll1986persistent}, this is called a \textit{linked data structure}). When using fat node persistence, each node in the pointer machine is made \emph{fat}, meaning that they not just store one value, but at most one value for each version that we keep track of. Whenever an update is performed on the node, the current value is not replaced, but a new value is added, including a version stamp. This way, each old version can still be accessed by searching for the corresponding version value in each node.



\subsection{Branch persistence}\label{sec:applying_branching_persistence}
Let $I(\S)$ be a dynamic pointer machine data structure storing the $n$ objects in $\S$. 
Instead of defining \emph{versions}, which are typically used in persistence, we will distinguish between the \emph{root} version of the data structure and different \emph{branches}. 
Each branch $i$ stores a subset $\S_i \subseteq \S$. A query then specifies a branch $i$, and performs the given query on the set $\S_i$. Similarly, an update specifies in which branch the update should be performed. To ensure that $\S_i$ remains a subset of $\S$, we only allow insertions of objects that are in $\S$. Additionally, one can perform a \emph{root-update} that performs the update on the root that stores the data structure on $\S$ \emph{and} then performs the same update on each branch. We illustrate this principle in Figure~\ref{fig:branching}.

\begin{figure}[h]
    \centering
    \includegraphics[page={1}]{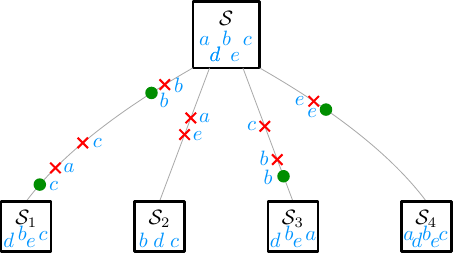}
    \caption{After creating four branches and making the indicated \emph{branch-updates} (each branch update either deletes or inserts the indicated object), the sets $\S_i$ consist of the indicated objects.}
    \label{fig:branching}
\end{figure}

\begin{definition}
    \label{definition:branch_persistence}
    Given a dynamic pointer machine data structure $I(\S)$ storing a (dynamic) set of $n$ objects~$\S$, a \emph{branch persistent} version is a data structure where:
    \begin{itemize}
        \item the \emph{root} (with label 0) corresponds to the current set $\S$,
        \item each \emph{branch}, with integer label $i$, corresponds to a subset $S_i \subseteq \S$.
    \end{itemize} Let $B$ denote the current number of branches. The following operations are supported:
    \begin{itemize}
    \item \emph{query}$(i, \rho)$: perform a query with $\rho$ on $I(\S_i)$.  
    \item \emph{branch-update}$(i, \sigma, \texttt{Insert} / \texttt{Delete})$: if $\sigma \in \S$ insert $\sigma$ into or delete $\sigma$ from $\S_i$.
        \item \emph{root-update}$(\sigma, \texttt{Insert} / \texttt{Delete})$: insert $\sigma$ into or delete $\sigma$ from $\S$, then perform the same update on each branch $i$ using a \emph{branch-update}. 
    \item \emph{branch}$(i)$: if there is no branch $i$, create a new branch with label $i$ that has $\S_i = \S$.
    \item \emph{rebuild}: rebuild the data structure from scratch. 
\end{itemize}
\end{definition}

\begin{lemma}\label{lemma:branch-persistent} A dynamic pointer machine data structure $I(\S)$ with $Q(n)$ query time and $U(n) \in \Omega(\log n)$ update time can be made branch persistent, supporting $B$ branches and with an upper bound of $N$ on $|\S|$ between two rebuilds, such that:
\begin{itemize}
    \item \emph{Queries} are supported in  $O(Q(N)\log B)$ time.
    \item  \emph{Root-updates} are supported in $O(U(N) B \log B)$ time and increase the space  by $O(U(N) B)$.
    \item \emph{Branch-updates} are supported in $O(U(N) \log B)$ time and increase the space  by $O(U(N))$. 
    \item \emph{Rebuilds} are supported in $O(( \,|\S|\, +z) U( \,|\S|\, ) \log B)$ time, where  $z$ is the size of the symmetric difference $z := \sum\limits_{i \in [B]} \, |\S - \S_i|$.  The space after the rebuild is $O((\,|\S|\,+z) U(\,|\S|\,))$. 
\end{itemize}
\end{lemma}
\begin{proof}
    We support branch persistence via the \emph{fat node persistence} technique from~\cite{driscoll1986persistent}.
The core idea is that each node in the pointer machine has up to $B$ versions, at most one per branch. These versions are labelled and stored in a balanced tree so that given any integer $i$ one can find the corresponding version of that pointer node in $O(\log B)$ time. 
This `version lookup' follows the following principle to access $I(\S_i)$: let $\nu$ be a pointer machine node, if there exists a version $\nu_i$ of $\nu$ then read this version, otherwise, read $\nu_0$.  We maintain the \emph{invariant} that reading the pointer machine in this manner recovers the data structure $I(\S_i)$ for all branches $i$.  As a consequence, queries take $O(Q(N) \log B)$ time. As \emph{branch}$(i)$ creates a new branch with $\S_i = \S$, this operation is supported in constant time by incrementing $B$.

\subparagraph{Updating the data structure.}
Immediately after creating a new branch $i$, the invariant holds as $\S_i = \S$. 
 We support the remaining updates as follows.

To support \emph{branch-update}$(i, \sigma, \texttt{Insert} / \texttt{Delete})$ we note that by our invariant, we can access $I(\S_i)$ with $O(\log B)$ overhead per operation.
Because the data structure $I(\S_i)$ can be updated in $U(N)$ time, such an update changes values (or pointers) in $U(N)$ nodes in the pointer structure. For every such change to a node $\nu$, we introduce (or overwrite) only $\nu_i$ representing its state in $I(\S_i)$. If $\nu$ did not exist before the update, we simply create a new node with version $\nu_i$. A \emph{branch-update} thus takes $U(N) \log B$ time and adds $U(N)$ space.
As the update on branch $i$ only changes $\nu_i$ for any node $\nu$, the invariant is immediately maintained for all other branches which use a different integer to read their node versions.

To support \emph{root-update}$(\sigma, \texttt{Insert} / \texttt{Delete})$ we execute  \emph{branch-update}$(i, \sigma, \texttt{Insert} / \texttt{Delete})$ for all $B$ branches, including the root branch, at a factor $B$ overhead. 



\subparagraph{Rebuilding.}
To support rebuilds we additionally maintain for all branches a balanced binary tree storing all objects in $\S$ that are \emph{not} in $\S_i$, i.e. $\S \setminus \S_i$, in some arbitrary fixed order. We call this tree the \emph{difference tree}.
When updating $\S_i$, the difference tree tree can be updated in $O(\log N)$ time and, as we assumed that $U(n) \in \Omega(\log n)$, this does not incur any overhead.
When we create a new branch, we create an empty difference tree in constant time.

A rebuild is executed as follows.
First, we construct given $\S$ the root data structure $I(\S)$ in $O(\,|\S|\, U(\,|\S|\,))$ time by performing an insertion for all  $\sigma \in \S$. 
Then, for each each old branch $i$, i.e. each branch that existed before the rebuild, we create a new branch in constant time. We then iterate over all elements $\sigma$ in the corresponding difference tree and perform  \emph{branch-update}$(i, \sigma, \texttt{Delete})$.  
Since the sum of the difference tree sizes equals $z$, the total reconstruction time is $O( (|\S| + z) U(|\S|) \log B)$ and the total space is $O((\,|\S|\,+z)U( \,|\S|\,))$.
\end{proof}

\subsection{Applying branch persistence}

We now apply branch persistence to the data structure maintaining our spanner, so that we no longer  store $O(\Psi^2 \eps^{-2})$ copies of the disk intersection data structure for each $\eps$-cell. Instead, we store a branch persistent data structure of each intersection query data structure. The data structure from Lemma~\ref{lemma:intersection_queries} for intersection queries can be made branch persistent using amortised expected update time and \emph{expected space} as follows.

\begin{lemma}\label{lemma:intersection_queries_branching} 
The data structure from Lemma~\ref{lemma:intersection_queries},  storing a dynamic set of disks $\S$ where between any two rebuilds $|\S| \leq N$, supports branch persistence such that: 
\begin{itemize}
    \item Queries are supported in $O(\log^2 N \log B)$ time.
    \item  Let $r$ and $b$ be the number of \emph{root-} and \emph{branch-updates} since the last rebuild and let  $z$ be the size of the symmetric difference $z := \sum\limits_{i} \, |\S - \S_i|$ at this rebuild. The rebuild and these updates take $O((\overline{n} + z + rB+b) \log^4 N \log  B)$ expected total time and the expected space usage is $O( (\overline{n} + z + rB + b) \log^4 N)$, where $\overline{n}$ denotes the size of $\S$ at the rebuild.
    
\end{itemize}
\end{lemma}
\begin{proof}
    We apply Lemma~\ref{lemma:branch-persistent} to the dynamic disk intersection data structure of Lemma~\ref{lemma:intersection_queries}. The disk intersection data structure has $Q(n) =  O(\log^2 n)$ query time  and $U(n) = O(\log^4 n)$ amortized expected update time. Lemma~\ref{lemma:branch-persistent} directly implies the stated query time of the branch persistent version of the data structure. The amortized expected insertion time of the data structure implies that $k$ updates on an initially empty data structure take $O(k \log^4 k \log B)$ total expected time. Since each update introduces space proportional to the original update time, the space usage is expected $O(k \log^4 k)$. The stated running time and space after a rebuild and some updates then follows from Lemma~\ref{lemma:branch-persistent}.
\end{proof}

\subparagraph{Auxiliary data structures.} In our original data structure, we stored two dynamic intersection data structures $Q(E,E')$ and $Q'(E,E')$ (see Definition~\ref{def:auxiliary_ds}) for \emph{every pair} of $\eps$-cells $(E,E')$, where $E'$ is bigger than $E$ and is ``close'' to $E$, where ``close'' is defined relative to only the size of $E'$. The structure $Q(E,E')$ stored all disks in the subpopulation $\Gamma_\eps(E)$ excluding disks that appeared in the maximal bichromatic matching $M(E,E')$, while $Q'(E,E')$ stored the disks in the population $\pi_\eps(E')$, again excluding any disks in $M(E,E')$. In our augmented data structure we will only store two data structures $\Q(E)$ and $\Q'(E)$ for each $\eps$-cell~$E$. These data structures are branch persistent, such that for each relevant pair $(E,E')$ we can store an additional branch (if necessary) instead of an almost identical copy of the data structure. Formally, we maintain the following structures.


\begin{definition}\label{def:auxiliary_ds_persistent}
For each $\eps$-cell $E$ in the  quadtree with non-empty population $\pi_\eps(E)$ we store the following two auxiliary data structures:
\begin{enumerate}
    \item $\Q(E)$ stores the subpopulation $\Gamma_\eps(E)$ in the branch persistent data structure of Lemma~\ref{lemma:intersection_queries_branching}.
    \item $\Q'(E)$ stores the population $\pi_\eps(E)$ in the branch persistent data structure of Lemma~\ref{lemma:intersection_queries_branching}.
\end{enumerate}
Furthermore, for every pair of $\eps$-cells $(E,E')$, where for their quadtree cells $C$ and $C'$ it holds that $|C'| = 2^{i-1} |C|$ for some $i \in [1,\lceil \log \Psi \rceil]$ and $C'$ intersects $(2^{i+5}+1) * C$, we:
\begin{enumerate}
    \item Store a maximal bichromatic matching $\MBM(E,E')$ of the bipartite graph $\Gamma_\eps(E) \times \pi_\eps(E')$, where there is an edge $(\sigma, \sigma') \in \Gamma_\eps(E) \times \pi_\eps(E')$ if and only if $\sigma$ and $\sigma'$ intersect. 
    \item If $\MBM(E, E')$ is non-empty, create a branch in $\Q(E)$ storing $
    \{ \sigma \in \Gamma_\eps(E) \mid \sigma \not \in \MBM(E, E') \}$.
    \item If $\MBM(E, E')$ is non-empty, create a branch in $\Q'(E')$ storing $\{ \sigma \in \pi_\eps(E') \mid \sigma \not \in \MBM(E, E') \}$.
\end{enumerate}
\end{definition}

 Note that, like with the original data structure, we risk (locally) having $O(\eps^{-2} \Psi^2)$ branches. However, the following lemma shows that this cannot occur too often, and that we can bound the \emph{global} symmetric difference between all persistent data structures and their branches.

\begin{lemma}\label{lemma:bounding_size_of_matching}
Let $\M$ be the set of all edges across all bichromatic matchings for all pairs of $\eps$-cells. Consider, for a fixed $\eps$-cell $E$, the set of disks $\S^{\Q(E)}$ and $\S^{\Q'(E)}$ stored in $\Q(E)$ and $\Q'(E)$ respectively and denote by
$z_E := \sum_i |\S^{\Q(E)} - \S^{\Q(E)}_i| + \sum_j |\S^{\Q'(E)} - \S^{\Q'(E)}_{j}| $ the sum over all symmetric differences across all branches of these data structures. Then we can upper bound the total symmetric difference as follows:
\[z := \sum_{\text{all } \eps\text{-cells } E} z_E  =  2 |\M| \in O(n \eps^{-2} \log \Psi) \]
\end{lemma}
\begin{proof}
Fix some $\eps$-cell $E$. For any branch $i$ of $\Q(E)$, the symmetric difference between the set $\S^{\Q(E)}$ and $\S^{\Q(E)}_i$  are exactly those disks that are in the matching $\M(E,E')$ where $E'$ is the $\eps$-cell corresponding to branch $i$. Similarly, the symmetric difference between $\S^{\Q'(E)}$ and $\S^{\Q'(E)}_{j}$ are exactly those disks that are in the matching $\M(E'',E)$ where $E''$ is the $\eps$-cell corresponding to branch $j$.
It follows that an edge $(\sigma,\sigma') \in \MBM(E,E') \subseteq \M$ is counted twice, once for $z_{E}$ and once for $z_{E'}$. So, $z$ equals $2|\M|$.

What remains is to upper bound $|\M|$. Consider a cell $C$ in the quadtree $T(\S)$ and an $\eps$-cell $E$ of $C$. 
We store edges across maximal bichromatic matchings between $\Gamma_\eps(E) \times \pi_\eps(E')$ for some `sufficiently close' $\eps$-cell $E'$. 
In particular, each disk in $\sigma$  appears in only one unique subpopulation $\Gamma_\eps(E)$ which corresponds to the $\eps$-cell $E$ of its storing cell $C$. There are subsequently $O(\log \Psi)$ cells $C'$ such that $|C'| = 2^{i-1}|C|$ for some $i \in [1,\lceil \log \Psi \rceil ]$ and $C' \subset (2^{i+5} + 1)  * C$. 
For every such cell $C'$, there are $O(\eps^{-2})$ $\eps$-cells $E'$ for which we consider a maximal bichromatic matching between $\Gamma_\eps(E) \times \pi_\eps(E')$.
It follows that each disk $\sigma$ is part of some ordered pair $(\sigma, \sigma') \in \mathcal{M}$ at most $O(\eps^{-2}\log \Psi)$ times which implies the lemma.
\end{proof}

\mainTwo*

\begin{proof}
    As before, we leave the generalisation to the case without bounding box to Appendix~\ref{app:generalisation}. We thus assume that at all times $\S$ remains contained in a bounding box $[0, \Psi^*]^2$ where $\Psi^*$ is the smallest power of two that is larger than $\Psi$.

    \subparagraph{Notation.} We repeat the notation introduced in the proof of Lemma~\ref{lemma:main_big_space}.
    Consider all quadtree cells that are at least unit size that can be obtained by recursively splitting the bounding box. 
    We denote by $\mathcal C$ the set of all pairs of such cells $(C,C')$ satisfying: $|C| = |C'|$ and $C' \subset 3*C$. We denote by $\mathcal E$ the set of all pairs of $\varepsilon$-cells $(E,E')$ satisfying: $E$ is a $\varepsilon$-cell of $C$, $E'$ is a $\varepsilon$-cell of $C'$, and there is an $i \in [1,\lceil \log \Psi \rceil]$ such that $|C'| = 2^{i-1} |C|$ and $C' \subset (2^{i+5}+1) * C$. We use $\MBM(E,E')$, $\Q(E)$ and $\Q'(E)$ as defined in Definition~\ref{def:auxiliary_ds_persistent}.

    \subparagraph{Disk insertions.} Inserting a disk $\sigma$ is done in the same way as in the original data structure. The difference is that we work on branches of the persistent data structure instead of copies of the disk intersection data structure. To be precise, we perform the following operations:
    \begin{enumerate}
        \item We first consider the storing family $\F(\sigma)$ of $\sigma$. For each cell $C \in \F(\sigma)$ we either identify the corresponding cell in $T(\S)$ or add this cell to $T(\S)$.
        \item 
        For all cells $C \in \F(\sigma)$, for every quadtree cell $C'$ such that $(C,C') \in \mathcal C$, we insert the centre of $\sigma$ into the Euclidean spanner $\mathit{Span}(C,C')$. 
        \item We consider the unique storing cell $C_\sigma$ of $\sigma$ and its $\eps$-cell $E$ that contains the centre of~$\sigma$. We perform a \emph{root-update}$(\sigma, \texttt{Insert})$ on $\Q(E)$ to insert $\sigma$ in the root version and all branches. We then consider all pairs  $(E,E') \in \mathcal E$.
        \begin{itemize}
        \item Consider the data structure $\Q(E)$.
        \begin{itemize}
            \item The $\eps$-cell $E'$ corresponds to a unique branch $i$ of $\Q(E)$.
        \end{itemize}
        \item Similarly, consider the data structure $\Q'(E')$.
        \begin{itemize}
            \item If the subpopulation $\Gamma_\eps(E)$ was empty before adding $\sigma$, then $\Q'(E')$ does not have a branch corresponding to $E$ yet. We add this branch and obtain its id $j$, storing it in the pair $(E, E')$.
            \item If $\Gamma_\eps(E)$ was not empty before, then we instead obtain the branch id $j$ from $(E, E')$.
        \end{itemize}
        \end{itemize}
        We then test whether $\sigma$ intersects a disk~$\sigma'$ in $\pi_\eps(E')$ that is not already in the maximal bichromatic matching  $\MBM(E,E')$ by querying the branch $j$ of $\Q'(E)$.
        \begin{itemize}
            \item If so, then we add $(\sigma, \sigma')$ to $\MBM(E,E')$. We delete $\sigma$ from the branch $i$ of $\Q(E)$, and delete $\sigma'$ from the branch $j$ of $\Q'(E')$ --- maintaining the invariant that these branches store all disks in $\Gamma_\eps(E)$ or $\pi_\eps(E')$ that do not appear in $\MBM(E,E')$.
        \end{itemize}
        \item To match the notation of Definition~\ref{def:auxiliary_ds_persistent}, we now introduce primes in our notation and consider all quadtree cells $C' \in \F(\sigma)$. A fixed cell $C'$ has a unique $\eps$-cell $E'$ which contains the centre of $\sigma$. We perform \emph{root-update}$(\sigma, \texttt{Insert})$ on $\Q'(E')$ to insert $\sigma$ in the root version and all branches. 
        We then consider all pairs  $(E,E') \in \mathcal E$:
        \begin{itemize}
        \item Consider the data structure $\Q(E)$.
        \begin{itemize}
        \item If the population $\pi_\eps(E')$ was empty before adding $\sigma$, then $\Q(E)$ does not have a branch corresponding to $E'$ yet. We add this branch and obtain its id $i$ and store it in the pair $(E, E')$.
            \item If $\pi_\eps(E')$ was not empty before, then we instead obtain the branch id $i$ from $(E, E')$.
        \end{itemize}
        \item Similarly, consider the $\eps$-cell $E'$ and  $\Q'(E')$.
        \begin{itemize}
                       \item Then $E$ corresponds to a unique branch $j$ of $\Q'(E')$ which we obtain from $(E, E')$.
        \end{itemize}
        \end{itemize}
        We test whether $\sigma$ intersects a disk~$\sigma''$ in $\Gamma_\eps(E)$ that is not already in the maximal bichromatic matching  $\MBM(E,E')$ by querying the branch $i$ of $\Q(E)$.
        \begin{itemize}
            \item If so, then we add $(\sigma'', \sigma)$ to the matching $M(E,E')$.
            We delete $\sigma$ from the branch $j$ of $\Q'(E')$, and we delete $\sigma''$ from the branch $i$ of $\Q(E)$ --- maintaining the invariant that these branches store all disks in $\Gamma_\eps(E)$ or $\pi_\eps(E')$ that do not appear in $M(E, E')$.
        \end{itemize}
        \item For every edge updated in any Euclidean spanner $\mathit{Span}(C,C')$, we update the corresponding type~\ref{type:one} edge in $G$. Similarly, whenever an empty bichromatic matching becomes non-empty, we add the corresponding type~\ref{type:two} edge to~$G$.
    \end{enumerate}
    
    \subparagraph{Disk deletions.} To delete a disk $\sigma$, we adapt the deletion in a similar way as the insertion by using the persistent data structures instead of the copies of the disk intersection data structure. We perform the following operations to delete $\sigma$:
    \begin{enumerate}
            \item We consider all quadtree cells, down to unit size, that are contained in $\sigma$ and contain the centre of $\sigma$. This is the storing family $\F(\sigma)$ of $\sigma$. If every cell $C \in T(\S)$ records $|\pi(C)|$, then we can simply iterate over all cells $C \in \F(\sigma)$ and decrement this counter.
        At the end of the procedure, we then remove all cells $C \in F(\sigma)$ for which we set the counter to zero, to update $T(\S)$. 
    \item We consider all quadtree cells $C$ that contain the centre of $\sigma$. For every quadtree cell $C'$ such that $(C,C') \in \mathcal C$, we delete the centre of $\sigma$ from the Euclidean spanner $\mathit{Span}(C,C')$.
    \item The disk $\sigma$ has a unique storing cell $C_\sigma$, which has a unique $\eps$-cell $E$ containing the centre of $\sigma$. We perform \emph{root-update}$(\sigma, \texttt{Delete})$ on $\Q(E)$ to delete $\sigma$ from the root version and all branches.
    We then consider all $E'$ such that $(E,E') \in \mathcal E$.
    For a fixed pair $(E, E')$ we note, just as with insertions, that the $\eps$-cell $E'$ corresponds to a unique branch $i$ of $\Q(E)$.
    Similarly, $E$ corresponds to a unique branch $j$ of $\Q'(E')$. We query whether $\sigma$ is in the maximal bichromatic matching $M(E,E')$.
    \begin{itemize}
        \item If there is a $(\sigma,\sigma') \in M(E,E')$, we delete $(\sigma,\sigma')$ from the matching. 
        To maintain a maximal matching, we need to test whether it is possible to reconnect $\sigma'$. We query the branch $i$ of $\Q(E)$ to test whether $\sigma'$ intersects a disk~$\sigma''$ in $\Gamma_\eps(E)$ that is not already in the maximal bichromatic matching  $\MBM(E,E')$.
        \begin{itemize}
            \item If so, then we add $(\sigma'', \sigma')$ to $M(E,E')$.
            We delete $\sigma''$ from the branch $i$ of $\Q(E)$ --- maintaining the invariant that the branch $i$ stores all disks in $\Gamma_\eps(E)$ that are not in the maximal matching $\MBM(E, E')$.
            \item Otherwise, we insert $\sigma'$ into the branch $j$ of $\Q'(E')$  --- maintaining the invariant that the branch $j$ stores all disks in $\pi_\eps(E')$ that are not in the maximal matching $\MBM(E, E')$. 
        \end{itemize}
    \end{itemize}
    \item
    To match the notation of Definition~\ref{def:auxiliary_ds_persistent}, we now introduce primes in our notation and consider all quadtree cells $C' \in \F(\sigma)$. 
    For a fixed $C'$, there is a unique $\eps$-cell $E'$ which contains the centre of $\sigma$. We perform \emph{root-update}$(\sigma, \texttt{Delete})$ on $\Q'(E')$ to delete $\sigma$ from the root version and all branches. 
    We consider all $E$ such that $(E,E') \in \mathcal E$. 
      For a fixed pair $(E, E')$ we note, just as above, that the $\eps$-cell $E'$ corresponds to a unique branch $i$ of $\Q(E)$.
    Similarly, $E$ corresponds to a unique branch $j$ of $\Q'(E')$.
    We query whether $\sigma$ is in the maximal bichromatic matching $M(E,E')$.
    \begin{itemize}
        \item If there is a $(\sigma'', \sigma) \in M(E,E')$, then we delete $(\sigma'', \sigma)$. To maintain a maximal matching, we need to test whether we can reconnect $\sigma''$. We query the branch $j$ of $\Q'(E')$ to test whether $\sigma''$ intersects a disk~$\sigma'$ in $\pi_\eps(E')$ that is not already in the maximal bichromatic matching  $\MBM(E,E')$.
        \begin{itemize}
            \item If so, we add $(\sigma'', \sigma')$ to $M(E,E')$. We delete $\sigma'$ from the branch $j$ of $\Q'(E')$ --- maintaining the invariant that the branch $j$ stores all disks in $\pi_\eps(E')$ that are not in the maximal matching $\MBM(E, E')$. 
            \item Otherwise, we insert $\sigma''$ into the branch of $i$ of $\Q(E)$ --- maintaining the invariant that the branch $i$ stores all disks in $\Gamma_\eps(E)$ that are not in  $\MBM(E, E')$. 
        \end{itemize}
    \end{itemize}
    \item For every edge updated in any Euclidean spanner $\mathit{Span}(C,C')$, we update the corresponding type~\ref{type:one} edge in $G$. Similarly, whenever a bichromatic matching~$M(E,E')$ becomes empty, we delete the corresponding type~\ref{type:two} edge from~$G$.
\end{enumerate}
Note that when we insert a disk in a specific branch of $\Q(E)$ or $\Q'(E')$, this disk is already present in the data structure, and thus present in the root version of the persistent data structure. This indeed allows us to use our defined \emph{branch-update} interface.

    \subparagraph{Global rebuild.} 
    We maintain a global counter which tracks the number of updates in $\S$. 
    When this counter hits a specific value $K$, we rebuild the branch persistent data structures $\Q(E)$ and $\Q'(E)$ for all $\eps$-cells $E$ (with a non-empty population $\pi_\eps(E)$) at the same time.  We set $K := N  / (2\Psi^2)$, where $N$ is $|\S|$ at the time of the last global rebuild. Note that when rebuilding like this, we maintain that the current number of disks $n$ is in $\Theta(N)$.
    
    \subparagraph{Update time.} We analyse the total time to perform a rebuild and the following $K$ updates. 
    
    A rebuild does not affect the Euclidean spanners, so we only need to consider the total update time for these spanners. Updating the Euclidean spanners $\mathit{Span}(C,C')$ and the type~\ref{type:one} edges of $G$ (step 2 and half of step 5), is done in the exact same manner as before. These updates thus take $O(K \cdot \eps^{-2}\log n \log \Psi \log^2(\eps^{-1}))$ total time (see the proof of Lemma~\ref{lemma:main_big_space}).

    Next, we analyse the time to rebuild and perform steps 3, 4, and 5 for the $K$ updates.
    Note that these $K$ updates may trigger branch and \emph{root-updates} in our branch persistent data structures. To bound the total running time, we will apply Lemma~\ref{lemma:intersection_queries_branching}, which requires the following input for each branch persistent data structure $\Q(E)$ (or $\Q'(E)$): 
    \begin{itemize}
        \item 
    $n_E$ as the number of disks stored in $E$ at the time of the rebuild,
    \item $z_E$ as the symmetric difference as defined in Lemma~\ref{lemma:bounding_size_of_matching},
    \item $r_E$ as the number of \emph{root-updates} performed on $\Q(E)$ and $\Q(E')$,
    \item  the number of branches $B$ which we always upper bound by $O(\frac{\Psi^2}{\eps^2})$,
    \item and $b_E$ as the number of \emph{branch-updates} performed on $\Q(E)$ and $\Q(E')$.
    \end{itemize}    
    The total time spent rebuilding and updating these data structures is then upper bound by expected:
    \[
    O\left( \sum_{\eps-\text{cells } E} (n_E + z_E + r_E \frac{\Psi^2}{\eps^2} + b_E) \log^4 n \log \Psi \log(\eps^{-1}) \right).
    \]

    In our proof, we split this sum into three sums and bound each sum separately. 
    To this end, we first observe that at the time of rebuild the total number of disks stored in these data structures is $O(n \log \Psi)$ and so:
    \[
    O\left( \sum\limits_{\eps-\text{cells } E} n_E \log^4 n \log \Psi \log(\eps^{-1}) \right) = O \left( n \log^4 n \log^2 \Psi \log \eps^{-1} \right). \]

    Moreover, we apply Lemma~\ref{lemma:bounding_size_of_matching}, to conclude that   $\sum_E z_E \in O(n \eps^{-2}\log \Psi)$ and so:
      \[
    O\left( \sum\limits_{\eps-\text{cells } E} z_E \log^4 n \log \Psi \log(\eps^{-1}) \right) = O \left( n \eps^{-2}  \log^4 n \log^2 \Psi \log \eps^{-1} \right). \]

    What remains is to bound the number of \emph{root}- and \emph{branch-updates} after $K = O(N / \Psi^2)$ updates to $\S$.
    Note that for each update in $\S$ step 3 and 4 execute the same steps, but the running time of step 4 dominates because it is executed for all $O(\log \Psi)$ cells in $\F(\sigma)$. 
    So, we consider inserting a disk $\sigma$ into $\S$ (or deleting a disk $\sigma$ from $\S$) and consider the operations performed by step 4.
    This step iterates over all $O(\log \Psi)$ cells in $\F(\sigma)$.     
    Such a fixed cell has a unique $\eps$-cell where we perform a \emph{root-update} on the $\Q'$ (of $\Q$ for step 3) data structure.
    It follows that after $K$ updates in $\S$, we can upper bound the total number of \emph{root-updates} by $\sum\limits_{\eps-\text{cells} E} r_E \in O(K \log \Psi)$.

    Our update procedure then continues, considering  $O(\Psi^2 / \eps^2)$ pairs $(E,E') \in \E$. For each pair, we perform a constant number of \emph{branch-updates} and queries in the corresponding $\Q$ and $Q'$ data structures. 
    It follows that after $K$ updates in $\S$, we can upper bound the total number of \emph{branch-updates} by $\sum\limits_{\eps-\text{cells} E} b_E \in O( \Psi^2 \eps^{-2} K \log \Psi)$.
    Since the branch update time dominates the time we spend on queries, we can upper bound the total expected time spent on \emph{root}- and \emph{branch-updates} after $K$ updates as follows:
        \[
    O\left( \sum\limits_{\eps-\text{cells } E} (r_E \eps^{-2} \Psi^2 + b_E) \log^4 n \log \Psi \log\eps^{-1} \right) = O \left( \Psi^2 \eps^{-2} K \log^4 n \log^2 \Psi \log \eps^{-1} \right). \]

    The running time is therefore dominated by the third sum. Amortised over $K$ updates, this yields an amortized expected running time of $O \left( \Psi^2 \eps^{-2} \log^4 n \log^2 \Psi \log \eps^{-1} \right).$ Note that the spanner update time carries a factor $\log^2 \eps^{-1}$, and so for ease of reading, we inflate the running time by a factor $\log \eps^{-1}$ and obtain the theorem. 

    \subparagraph{Space usage.} The space usage of the Euclidean spanners and the bichromatic matchings is equivalent to those in Lemma~\ref{lemma:main_big_space}. The analysis of the space usage of the $\Q$ and $\Q'$ data structures is analogous to the analysis of the update time, except that the result is a factor $O( \log  \Psi \log(\eps^{-1}))$ smaller because of the $O(\log B)$ difference between the update time and space in Lemma~\ref{lemma:intersection_queries_branching}. It follows that the space usage of our data structure is expected $O((n\eps^{-2} \log \Psi + K \Psi^2 \eps^{-2} \log \Psi) \log^4 n)$. As we rebuild after $K = N /\Psi^2$ updates, the expected space is bounded by $O(n\eps^{-2} \log^4 n \log \Psi)$.
    \end{proof}

\section{Removing the bounding box assumption}
\label{app:generalisation}

Throughout the paper we assumed that $\S$ is a fully dynamic set of $n$ objects (disks in the plane or $d$-dimensional hypercubes in $\mathbb{R}^d$) that at all times lie in a bounding box $B = [0, \Psi^*]^d$ where $\Psi^*$ is the smallest power of two bigger than $\Psi$. 
In this section, we explain the reasoning for this assumption and how to remove it (allowing $\S$ to lie arbitrarily in $\mathbb{R}^d$). 

\subparagraph{Static quadtrees and compression.}
Our core approach is to store a set of objects (disks or $d$-dimensional hypercubes) $\S$ in a dynamic $d$-dimensional quadtree. A quadtree takes as input a bounding hypercube $B$. For any disk $\sigma \in \S$, we recursively subdivide $B$ until we find the first hypercube $C$ that contains the centre of $\sigma$ and is fully contained in $\sigma$. However, if $B$ is unbounded, then the number of required subdivisions to create or identify such a hypercube $C$ is itself unbounded.

In the static setting, one can compute a bounding hypercube $B$ of $\S$ of finite size in linear time. However, even if all objects in $\S$ have bounded size, the size of $B$ may still be arbitrarily large, as the objects may be far apart. Since all objects in $\S$ have diameter at least~$4$, their storing cells have side length at least $1$. Consequently, the number of subdivisions required to create or identify any storing cell is $\Theta(\log |B|)$, where $|B|$ denotes the side length of the bounding hypercube. A naive quadtree storing $n$ objects in $B$ therefore requires $O(2^d n \log |B|)$ space and can be constructed in $O(2^d n \log |B|)$ time.

In the static setting, this spatial overhead can be avoided by using a \emph{compressed quadtree}~\cite{har-peled2011geometric}. While the standard definition applies to a quadtree that stores a set of $n$ points, we paraphrase it to its natural extension that stores a set of $n$ objects as follows. We start with a bounding hypercube $B$ and subdivide it into its $2^d$ children. If there exist two objects in $\S$ whose centres lie in different children of $B$, then we partition $\S$ according to these children and recurse. Otherwise, all centres lie in a single child, and we apply \emph{compression}. Ideally, we would identify the largest descendant $C'$ of $B$ for which this condition no longer holds and make $C'$ a child of $B$. However, under the real RAM model, where centres may have arbitrary real coordinates, such a cell $C'$ cannot in general be computed, as this requires integer division. Instead, we define $C'$ to be the minimum bounding hypercube containing all centres in $\S$. We then make $C'$ the (compressed) child of $B$ and recurse on $C'$, which, by construction, does not trigger further compression. The resulting tree has size $O(n)$ and, for constant dimension $d$, can be constructed in $O(n \log n)$ time.

\subparagraph{Dynamic misalignment.}
There is currently no known algorithm for dynamically maintaining a compressed quadtree, whether storing a set of points, disks, or more general geometric objects. Two main complications arise when attempting to maintain such a structure dynamically.
The first issue concerns updates that affect the bounding hypercube $B$. As $\S$ changes, so may its bounding box. When inserting an object whose centre lies outside $B$, this can be handled straightforwardly: for constant dimension $d$, we can compute in constant time a new bounding hypercube $B'$ that contains both $B$ and the new object, and make $B$ a compressed child of $B'$. Similarly, upon deletion, the bounding hypercube can only shrink if the root has a compressed child, in which case this child becomes the new root.

\begin{figure}
    \centering
    \includegraphics{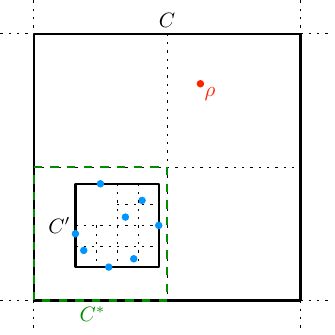}
    \caption{When the red centre point is inserted, the entire subtree of the compressed child $C'$ has to be realigned.}
    \label{fig:misaligned}
\end{figure}

The second issue is more fundamental, and we call it the problem of \emph{misalignment}. Consider a quadtree cell $C$ that has a compressed child $C'$. The cell $C$ has a compressed child because splitting $C$ normally would result in all centres of disks in $\S$ lying in a single child $C^*$ of $C$ (see Figure~\ref{fig:misaligned}).  By definition, $C'$ is the minimum-size bounding hypercube containing all centres in $\S$. Now suppose that the subtree rooted at $C'$ is large and contains no further compression.
If we now insert a new disk whose centre $\rho$ does not lie inside $C^{*}$, the compression invariant at $C$ is violated. In particular:
\begin{itemize}
\item $C$ no longer satisfies the condition for having a compressed child and should instead have $2^d$ children, one of which is $C^*$;
\item if we then attempt to recursively rebuild the quadtree starting from $C^*$ according to the compression rules, we encounter the problem that the resulting subtree is entirely different from the existing subtree rooted at $C'$ (thus, this takes $O(n)$ time). 
\end{itemize}

We refer to this phenomenon as \emph{misalignment}, as the quadtree rooted at $C'$ is not aligned with the quadtree that would arise from recursively subdividing $C^*$. The issue stems from the specific invariant used to define compression. In the fully dynamic setting, no suitable invariant is known that avoids this problem.
Although this particular example has a trivial amortisation argument, there is no amortised (or worst-case) known solution to this overall problem. 
In particular, even in dimension $d = 2$, it is unknown how to maintain a linear-size quadtree $T$ storing a set of geometric objects that simultaneously satisfies the following:
\begin{itemize}
\item each node $C$ in $T$ has either a single compressed child or is partitioned into $O(2^d)$ children;
\item each quadtree cell is contained in all of its ancestor cells;
\item each point in the $\mathbb{R}^d$ is assigned to exactly one quadtree leaf; and
\item no two quadtree cells overlap if their sizes differ by less than a factor of $2$.
\end{itemize}

These properties are, to the best of our knowledge, necessary for using a quadtree as a black-box data structure in geometric algorithms. Moreover, we are not aware of any subset of these properties that is both sufficient for such applications and achievable without encountering the misalignment problem.

\subsection{A static well-defined $(1+\eps)$-spanner}

We present a domain-specific solution to the misalignment problem. Specifically, we introduce a two-layer quadtree tailored to the maintenance of our spanner (and later connectivity) data structures. The key idea is to separate the problem into local structures which have fixed-size bounding boxes.
We then heavily exploit amortisation to maintain these two layers. 

Our data structure consists of what we call \emph{focused spanners}.
A focused spanner $G$ is defined with respect to a centre point $c$, and it has some associated set of objects $\S_G \subseteq \S$. Each focused spanner is assigned a unique ID. It is associated with two regions: a \emph{focal area}, defined as the hypercube $[0, \frac{1}{2}\Psi]^d$ centred at $c$, and a \emph{connection area}, 
defined as the hypercube $[0, 6\Psi]^d$ centred at $c$. 
In addition, we maintain a collection $T_G$ of $(1+\eps)$-spanners, where each spanner in $T_G$ is a $(1+\eps)$-spanner defined by two things: the bounding box $[0, 18\Psi]^d$ centred at $c$ and an intersection graph 
$D(\S')$ where $\S_G \subseteq \S' \subseteq \S$.
The top-level structure is a collection of focused spanners:

\begin{figure}
    \centering
    \includegraphics{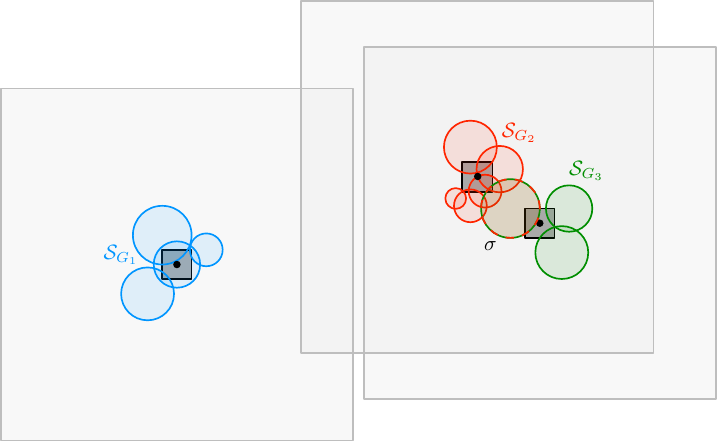}
    \caption{A focused spanner decomposition $\G$ of $\S$ with three focused spanners $G_1,G_2,G_3$ with their centre, focal area (black), and connection area (grey). Disk $\sigma$ is in both $\S_{G_2}$ and $\S_{G_3}$. The set $T_{G_1}$ contains a spanner on $\S_{G_1} \cup \S_{G_2}$, $T_{G_2}$ contains a spanner on $\S_{G_2} \cup \S_{G_3}$, and $\T_{G_3} = \emptyset$.}
    \label{fig:focused_spanner}
\end{figure}

\begin{definition}[Figure~\ref{fig:focused_spanner}]
\label{def:focused_spanner_def}
A \emph{focused spanner decomposition} $\mathcal{G}$ of $\S$ is a collection of focused spanners (each with a unique ID $i$) satisfying the following properties:
\begin{itemize}
\item For any distinct $G_1, G_2 \in \mathcal{G}$, the focal areas of $G_1$ and $G_2$ are disjoint;
    \item For every $\sigma \in \S$, there exists at least one $G \in \mathcal{G}$ such that $\sigma$ intersects the focal area of $G$;
    \item For each $G \in \mathcal{G}$, we maintain a set $\S_G$ containing a copy of every element of $\S$ that intersects its focal area. Moreover, we store $O(1)$ copies $\S_G^1, \S_G^2, \ldots$ of $\S_G$, where the constant depends only on the $d$-dimensional packing number;
    \item Consider any $G_1, G_2 \in \mathcal{G}$ (possibly $G_1 = G_2$) with centres $c_1$ and $c_2$ such that their connection areas intersect, and assume that the ID of $G_1$ is smaller than that of $G_2$. Let $B$ denote the hypercube $[0, 18\Psi]^d$ centred at $c_1$. Then all objects in $\S_{G_1} \cup \S_{G_2}$ are contained in $B$, and we require that one of the spanners in $T_{G_1}$ is a $(1+\eps)$-spanner constructed on the fixed bounding box $B$ with input $\S' = \S_{G_1}^i \cup \S_{G_2}^j$, where $\S_{G_1}^i$ and $\S_{G_2}^j$ are arbitrary copies as defined above.
\end{itemize}

\end{definition}

\begin{lemma}
    \label{lem:data_structure_size}
    Consider any data structure that, for a set of objects $\S'$ within a bounding box~$B$, stores a $(1+\eps)$-spanner of the intersection graph of $\S'$.
    The focused spanner decomposition using this data structure to store a set of objects $\S$ has the same asymptotic size and space usage. Moreover, the union of all focused spanners in the decomposition is a $(1+\eps)$-spanner of $\S$.  
\end{lemma}

\begin{proof}
Assume that we are given a focused spanner decomposition $\mathcal{G}$ of $\S$ constructed using the data structure from the lemma statement. We first argue this gives a $(1+\eps)$-spanner.
Consider any two objects $\sigma_1, \sigma_2 \in \S$ that intersect. Then one of the following cases holds:
\begin{enumerate}
    \item $\sigma_1$ and $\sigma_2$ both intersect the focal area of some focused spanner $G \in \mathcal{G}$; or
    \item $\sigma_1$ and $\sigma_2$ intersect the focal areas of two focused spanners $G_1, G_2 \in \mathcal{G}$, respectively, whose connection areas intersect.
\end{enumerate}
In both cases, by Definition~\ref{def:focused_spanner_def}, there exists a focused spanner $G^* \in \mathcal{G}$ and a $(1+\eps)$-spanner $T_{G^*}$ of $\D(\S')$ where both $\sigma_1$ and $\sigma_2$ are in $\S'$.
Since $\sigma_1$ and $\sigma_2$ intersect, this spanner must contain a path between these two of length at most $(1+\eps)\, d(\sigma_1, \sigma_2)$. We can apply this argument to all edges $(\sigma_1, \sigma_2)$ along a path in $\D(\S)$ between any two objects $\sigma$ and $\rho$, and it follows that the union of all spanners in our focused spanner decomposition contains a path from $\sigma$ to $\rho$ whose length is at most a factor $(1+\eps)$ larger than the shortest path in $\D(\S)$. Hence, the union of all spanners is a $(1+\eps)$-spanner for $\S$.

We now analyse the spanner size and data structure space usage. For constant dimension~$d$, the packing number is constant, implying that any object $\sigma \in \S$ intersects only constantly many focal areas. Furthermore, each focused spanner intersects the connection area of only constantly many other focused spanners. Therefore, each $\sigma \in \S$ participates in only constantly many $(1+\eps)$-spanners. It follows that the total size of all maintained spanners, as well as the overall space usage, increases only by a constant factor.
\end{proof}

\subsection{Maintaining our spanner}

We show that a focused spanner decomposition can be maintained dynamically by relying on standard amortised techniques. Instead of using \emph{static} $(1+\eps)$-spanners, we assume that these spanners are dynamic, i.e. allow for objects to be inserted and deleted.
Furthermore, we rebuild the entire data structure after $N$ updates, or whenever $|\S| = \tfrac{1}{2}N$ or $|\S| = 2N$, where $N$ is $|\S|$ at the most recent rebuild. It follows that rebuilds occur only after $\Theta(N)$ updates and that, at all times, $|\S| \in [\tfrac{1}{2}N, 2N]$.

By Lemma~\ref{lem:data_structure_size}, a focused spanner decomposition incurs no asymptotic overhead in input size, as each object in $\S$ is stored only $O(1)$ times. Consequently, a rebuild can be performed using $O(N)$ data structure insertions. These can be charged to the $\Theta(N)$ updates that triggered the rebuild, implying that the rebuilding scheme introduces no asymptotic overhead in either running time or space usage.
It remains to describe how to maintain a focused spanner decomposition between rebuilds. Let $\mathcal{G}$ denote the current decomposition. We maintain the invariant that $|\mathcal{G}| \leq 2N \in \Theta(|\S|)$ at all times. The focal areas of the spanners in $\mathcal{G}$ are stored in a dynamic hypercube intersection data structure, and we maintain $\S$ in a similar structure. Since any dynamic spanner data structure must inherently support $d$-dimensional hypercube intersection queries (or, worse, disk intersection queries), this introduces no asymptotic overhead.

\subparagraph{Supporting deletions.}
Deletions are straightforward under our amortisation scheme. When deleting an object $\sigma \in \S$, we observe that $\sigma$ is stored in only $O(1)$ spanner data structures and sets $\S_{G}^i$. We remove $\sigma$ from each of these structures, incurring no asymptotic overhead. We do not update the set of focused spanners in $\mathcal{G}$ during deletions, and so there can exist focused spanners $G \in \mathcal{G}$ where $\S_{G}$ is empty. 

\subparagraph{Supporting insertions.}
When inserting an object $\sigma$ into $\S$, we first query the intersection data structure to identify the $O(1)$ focused spanners $G \in \mathcal{G}$ whose focal areas intersect $\sigma$. For each such $G$, we insert $\sigma$ into all $O(1)$ maintained copies $\S_G^1, \S_G^2, \ldots$ of $\S_G$. Some of these copies $\S_G^{i}$ participate in $(1+\eps)$-spanner data structures on $D(\S')$ with $\S' = \S_G^{i} \cup \S_{G'}^{j}$ for some other focused spanner $G' \in G$ whose connection area intersects that of $G$.
By maintaining a pointer between $\S_G^{i}$ and the sets $\S'$ that it is a part of, we can find these in constant time. 
We insert $\sigma$ into each of these $O(1)$ $(1+\eps)$-spanners with a constant factor overhead. 

Next, we consider the hypercube $F = [0, \tfrac{1}{2}\Psi]^d$ centred at the centre of $\sigma$. We use our hypercube intersection query data structure to query whether there exists a focused spanner $G' \in \mathcal{G}$ whose focal area intersects $F$. If such a focused spanner $G'$ exists, then we already maintain the invariant that each object in $\S$ intersects one focal area of a focused spanner in $\mathcal{G}$. Otherwise, if no such $G'$ exists, we create a new focused spanner $G$ with focal area $F$ and assign it a unique ID. 
We then have to update $\mathcal{G}$. We first compute the set $\S_G$ by reporting all objects in $\S$ that intersect $F$ using a hypercube intersection query. This has an additive overhead of $O(k)$, where $k$ denotes the number of objects in $\S$ that intersect $F$. We then construct the $O(1)$ copies $\S_G^1, \S_G^2, \ldots$ in $O(k)$ time. 
Next, we identify all focused spanners $G' \in \mathcal{G}$ whose connection areas intersect those of $G$, again using a hypercube intersection query. For each such $G'$, we select unused copies $\S_G^i$ and $\S_{G'}^j$, and build a new $(1+\eps)$-spanner on $\S' = \S_G^i \cup \S_{G'}^j$ using the bounding box associated with the spanner of smaller identifier, as required by Definition~\ref{def:focused_spanner_def}. Repeating this for all such $G'$ restores the focused spanner decomposition.
We note that this latter construction step can be expensive: firstly, the factor $O(k)$ may be considerable. Even worse, there may be $O(1)$ focused spanners $G'$ whose connection area intersects that of $G$, and each of these $G'$ may store a considerable number of objects. 
However, we argue that all this time can be charged to an amortisation scheme. 
Indeed, all objects in $\S$ have diameter at most $\Psi$, and at all times $|\S| \in \Theta(N)$. Furthermore, focal areas are only created and never deleted between rebuilds. By a packing argument, each object in $\S$ intersects only $O(1)$ focal areas between any two rebuilds. Consequently, each object participates in only $O(1)$ newly constructed $(1+\eps)$-spanners during $\Theta(N)$ updates.
We can therefore charge the total cost of these rebuild operations to the $\Theta(N)$ updates between rebuilds, which implies that insertions incur no asymptotic overhead.

\subparagraph{Applying the focused spanner decomposition to our spanners.} The dynamic focused spanner decomposition can be directly applied to our dynamic $(1+\eps)$-spanners to lift the restriction that all objects should be contained in a bounding box $B = [0, \Psi^*]^d$ where $\Psi^*$ is the smallest power of two that is bigger than $\Psi$. The only caveat is that for the spanners of the focused spanner decomposition, we cannot restrict the bounding size to be a power of two. However, this can be fixed by relaxing the requirement that the smallest quadtree cells have unit size, but instead letting them have size $u$, where $u$ is the size of the first quadtree cell of at most unit size that is obtained by recursively splitting the given bounding box.

\section{Connectivity}\label{sec:connectivity}

The input is again a dynamic set of disks $\S$ where every disk has diameter between $4$ and a fixed and known $\Psi$.
The connectivity data structure allows for both insertion and deletion of disks, and furthermore, it can answer queries of the form: Given a pair of disks $\sigma, \rho \in \S$, do they belong to the same connected component in the disk graph $\D(\S)$?

As previously stated in Section~\ref{sec:overview}, by plugging in our dynamic spanner into the dynamic connectivity data structure for general graphs~\cite{holm2001poly}, we obtain a connectivity data structure for dynamic disk graphs. Compared to the state-of-the-art~\cite{baumann2024dynamic}, this plug-in approach gives better bounds on the space, i.e. logarithmic in $\Psi$ compared to linear~\cite{baumann2024dynamic}, but gives worse bounds on the update time, i.e. quadratic in $\Psi$ compared to linear~\cite{baumann2024dynamic}.

In this section, we obtain the best of both worlds, i.e. a dynamic connectivity data structure matching (i) the space of the plug-in approach and (ii) the update time of~\cite{baumann2024dynamic}, up to logarithmic factors. To achieve these bounds, our approach is to combine and incorporate the techniques from our dynamic spanner into the dynamic connectivity pipeline of~\cite{baumann2024dynamic, hoog2024fully}. We will formally introduce the connectivity pipeline in Section~\ref{sec:connectivity_pipeline}. 


At a high level, the connectivity pipeline of Baumann, Kaplan, Klost, Knorr, Mulzer, Roditty, and Seiferth~\cite{baumann2024dynamic} combines a quadtree with four data structures to answer connectivity queries. 
This construction was then modified by van der Hoog, Nusser, Rotenberg, and Staals~\cite{hoog2024fully} (we forward reference Figure~\ref{fig:pipeline}).
In this section, we note that the bottleneck for this pipeline is the maintenance of maximal bichromatic matchings for bichromatic intersection graphs between quadtree cells $(C, C')$.
We observe that this task is very similar to the maintenance of a maximal bichromatic matching between $\eps$-cells $(E, E')$ in Section~\ref{sec:high_space} and Appendix~\ref{sec:small_space}.
In fact, we show that we can replace the disk intersection data structure in the pipeline with our branch persistent variant to obtain a better space bound.

\subsection{Bounding box assumption}

Note that by applying the same technique as in Appendix~\ref{app:generalisation}, we do not \emph{have} to assume that~$\S$ remains contained in a fixed bounding box $[0, \Psi^*]^2$. However, for ease of exposition, we do make this assumption so that we can store $\S$ in a singular uncompressed quadtree $T(\S)$. 
We note that in previous sections, we defined $T(\S)$ as the uncompressed quadtree that stores for each disk $\sigma \in \S$ the \emph{storing family}.
For maintaining connectivity, it suffices to use the minimum-size quadtree that contains only the storing cells $C_\sigma$ for $\sigma \in \S$. With slight abuse of notation, we call a cell $C$ a storing cell if there is a $\sigma \in \S$ with $C$ as its storing cell. We warn the reader that, henceforth, we overload the notation $T(\S)$:

\begin{definition}
    \label{def:quadtree_sparse}
    Given a fixed bounding box $B = [0, \Psi^*]^2$ and the set of disks $\S$, we define $T(\S)$ as the minimum-size quadtree that contains for all $\sigma \in \S$ the storing cell $C_\sigma$ (the largest cell $C$ that can be obtained by recursively splitting $B$, that is contained in $\sigma$).
    We define, for any $C \in T(\S)$, its \emph{garrison} $\Gamma(C)$ to be the set of all disks in $\S$ that have $C$ as their storing cell.     
\end{definition}

\subsection{The existing pipeline}
\label{sec:connectivity_pipeline}

We first formally define the existing pipeline and then explain which data structures to replace within this pipeline. 
To this end, we introduce two important definitions from~\cite{hoog2024fully}:

\begin{definition}[Based on definitions in~\cite{hoog2024fully}]\label{def:connectivity_sets}
    Let $\sigma \in \S$ be a disk.
    \begin{itemize}
        \item The constituents $\C(\sigma)$ of $\sigma$ consists of all cells $C$ such that  (see Figure~\ref{fig:illustratingC}):
        \begin{enumerate}
            \item $C$ is contained in $\sigma$, and
            \item either $\Gamma(C) \neq \emptyset$ or a descendant $C'$ of $C$ has $\Gamma(C')\neq \emptyset$, and
            \item the parent of $C$ is not contained in $\sigma$.
        \end{enumerate}
        \item The \emph{perimeter} $\P(\sigma)$ of $\sigma$ consists of all cells $C$ with diameter at most $|\sigma|$ such that $(7 * C)$ intersects the \emph{boundary} of $\sigma$ (see Figure~\ref{fig:illustratingP}).  
    \end{itemize} 
\end{definition}

\begin{figure}
    \centering
    \includegraphics[page=4]{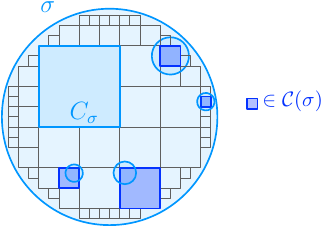}
    \caption{For a disk $\sigma$, the maximal quadtree cells (in the complete quadtree) contained in $\sigma$ and in dark blue the subset that are the constituents $\mathcal{C}(\sigma)$.}
    \label{fig:illustratingC}
\end{figure}

\begin{figure}
    \centering
    \includegraphics[page=5]{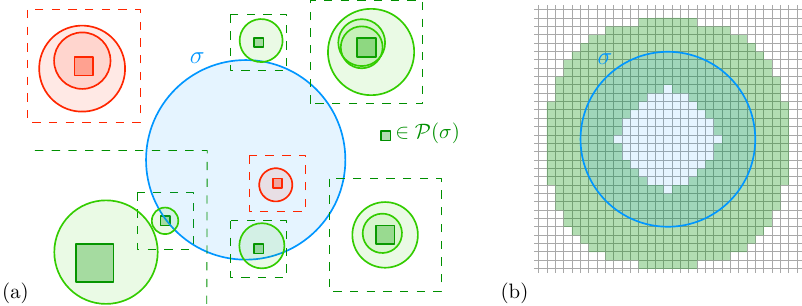}
    \caption{(a) For a disk $\sigma$, the storing cells in the perimeter $\mathcal{P}(\sigma)$ in green, the storing cells that are not in the perimeter $\mathcal{P}(\sigma)$ in red, and the dashed boundary of $(7 * C)$ for each of the shown cells. (b) On a fixed quadtree level, the green cells are potentially on the perimeter of $\sigma$.}
    \label{fig:illustratingP}
\end{figure}


\begin{definition}
    For $C \in T(\S)$, we define $\C(C) := \bigcup_{\sigma \in \Gamma(C)}\C(\sigma)$ and $\P(C) := \bigcup_{\sigma \in \Gamma(C)}\P(\sigma)$.
\end{definition}

\subparagraph{The proxy graph and connectivity.}
At a high-level, connectivity in the connectivity pipeline is maintained via a \emph{proxy graph}.
This is a graph where the vertices are the quadtree cells.
There exists an edge between quadtree cells $(C, C')$ if and only if $C'$ lies in the perimeter $\mathcal{P}(C)$ of $C$ \emph{and} there exists an intersection between two disks in the bichromatic intersection graph $\Gamma(C) \times \Gamma(C')$. 
By marking for every cell $C$, all cells $C' \in \mathcal{C}(C)$, any connectivity query between disks $\sigma \in \Gamma(C_1)$ and $\rho \in \Gamma(C_2)$ can then be answered by querying connectivity in the proxy graph between the two highest marked ancestors of $C_1$ and $C_2$.

Formally, the pipeline from~\cite{hoog2024fully} has five components,  see Figure~\ref{fig:pipeline}.
We describe the pipeline using our definitions, and assume that $\S$ is a set of disks instead of squares:\footnote{In particular, see the pipeline of Section 8 in~\cite{hoog2024fully}. Note that items 3 and 4 are reversed w.r.t.~\cite{hoog2024fully}.}

\begin{enumerate}
    \item\label{pipeline1} We store all disks in $\S$ in a quadtree $T(\S)$ from Definition~\ref{def:quadtree_sparse}.
    \item\label{pipeline2} For each storing cell $C \in T(S)$, we \emph{mark} all quadtree cells $C' \in \C(C)$ in the marked-ancestor tree (MAT) by Alstrup, Husfeldt, and Rauhe~\cite{AlstrupHR98}.
    \item\label{pipeline3} For each storing cell $C \in T(\S)$, we maintain for every storing cell $C'$ in its perimeter $\mathcal{P}(C)$ a maximal bichromatic matching $\MBM(C, C')$ in the bichromatic graph $\Gamma(C) \times \Gamma(C')$.
    \item\label{pipeline4} For each storing cell $C \in T(\S)$, we maintain for every storing cell $C'$ in its perimeter $\mathcal{P}(C)$ two disk intersection data structures:
    \begin{itemize}
        \item A disk intersection data structure storing all disks in $\Gamma(C)$ that are not in  $\MBM(C, C')$, 
        \item A disk intersection data structure storing all disks in $\Gamma(C')$ that are not in $\MBM(C, C')$.
    \end{itemize}
    \item\label{pipeline5} We store the proxy graph where the vertex set consists of the quadtree cells in $T(\S)$ and there is an edge between quadtree cells $(C, C')$ if their maximal bichromatic matching $\MBM(C, C')$ is well-defined and non-empty. 
    All edges of this graph are stored in the connectivity query data structure by Holm, de Lichtenberg, and Thorup~\cite{holm2001poly}.
\end{enumerate}

\begin{figure}
    \centering
    \includegraphics[page=2]{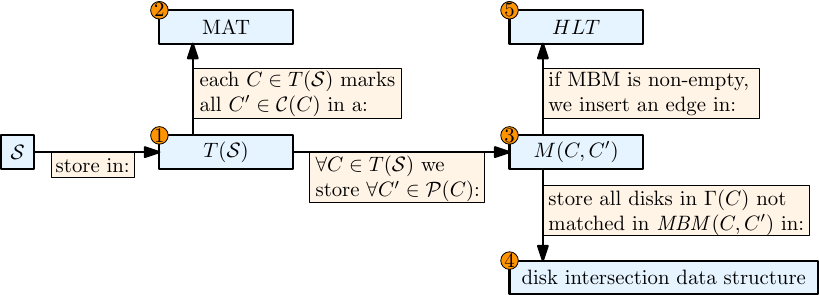}
    \caption{An overview of the pipeline used in~\cite{baumann2024dynamic, hoog2024fully}. The quadtree $T(\S)$ stores all storing cells. For a cell $C \in T(\S)$ we consider the constituents $\C(C) \subseteq T(\S)$ and the perimeter $\P(C)\subseteq T(\S)$.}
    \label{fig:pipeline}
\end{figure}

\subsection{Modifying and implementing the pipeline}

First, we note that items~\ref{pipeline1}, \ref{pipeline3} and~\ref{pipeline5} of this pipeline are shape-independent, i.e. they work the same for disks and squares. 
They are also identical across the pipelines from~\cite{baumann2024dynamic} and~\cite{hoog2024fully}.
We do not need to modify these item definitions, nor their dynamic maintenance. 

Items~\ref{pipeline2} and~\ref{pipeline4} are shape-dependent; the definitions of the sets $\mathcal{C}(\sigma)$ and $\mathcal{P}(\sigma)$ change based on whether $\sigma$ is a disk or a square. Moreover, there is a big difference between disk and square intersection data structures. 
As a consequence, items~\ref{pipeline2} and~\ref{pipeline4} are defined slightly differently between the original pipeline of~\cite{baumann2024dynamic} and the pipeline from~\cite{hoog2024fully}.
Even though our setting is more similar to~\cite{baumann2024dynamic}, our techniques will be more similar to~\cite{hoog2024fully}. For ease of exposition, we thus fully adopt the above pipeline from~\cite{hoog2024fully} and show how to maintain item~\ref{pipeline2} and~\ref{pipeline4} under insertions and deletions of disks. 
We wish to provide maximum credit to~\cite{baumann2024dynamic}, and note that their approach would suffice to maintain item~\ref{pipeline2}. We only show how to maintain item~\ref{pipeline2} for completeness.

\subparagraph{Maintaining pipeline item~\ref{pipeline2} and~\ref{pipeline4}.}
We continue by showing how to maintain pipeline items~\ref{pipeline2} and~\ref{pipeline4} for a set of disks $\S$ with their diameters in $[4, \Psi]$. 
To this end, we first recreate the lemmas from~\cite{hoog2024fully} that bound the size of the constituents $\mathcal{C}(\sigma)$ and the perimeter $\mathcal{P}(\sigma)$, for disks.

\begin{lemma}[Lemma 4.4 in~\cite{baumann2024dynamic}]\label{lemma:bounding_size_of_C_and_P}
    For each $\sigma \in \S$ there are $O(\Psi)$ cells in $|\C(\sigma)|$ and~$|\P(\sigma)|$.
\end{lemma}

\begin{figure}
    \centering
    \includegraphics[page=2]{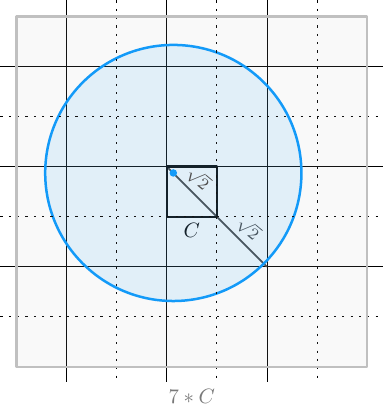}
    \caption{Illustrating why $7*C$ is needed to contain a disk stored in a storing cell $C$. In this figure, the storing cell $C$ has unit size.}
    \label{fig:whycandPneedsSevenNeighb}
\end{figure}

\begin{lemma}[Lemma 6 in~\cite{hoog2024fully} ]\label{lem:in_perimeter}
    Let $C$ be a storing cell. Denote by $Z$ any cell, such that there could exist a disk $\sigma$ with storing cell $Z$ and $C \in\P(\sigma)$. There are $O(\log \Psi)$ such cells. 
\end{lemma}
\begin{proof}
    Let $\ell$ be the level of the storing cell $C$; that is, $|C| = 2^\ell$. By the definition, every $Z$ such that there possibly exists a $\sigma$ with storing cell $Z$ such that $C\in \P(\sigma)$ must be on the same level as $C$ or on a level $j \geq \ell -1$. At level $j \geq \ell$ in the quadtree, there is an ancestor of $C$, $C_j$, with size $2^j$. As $Z$ is the storing cell of $\sigma$, the diameter of $\sigma$ is less than $4\sqrt{2}|Z|$ and thus $\sigma$ is contained in $7 * Z$ (see Figure~\ref{fig:whycandPneedsSevenNeighb}). If $C \in \P(\sigma)$, then by definition the boundary of $\sigma$ intersects $(7 * C)$. In particular, this means that $(7 * C_j)$ and $(7 * Z)$ must overlap. Equivalently, the cell $Z$ needs to intersect $(49 \,* C_j)$. Therefore at every level $j \geq \ell$ there are only  $O(1)$ number of candidates for $Z$.  For level $j = \ell -1$ there are even fewer such cells than for level $j=\ell$. We conclude that there are at most $O(\log \Psi)$ cells with the properties of~$Z$.
\end{proof}

Having bounded the complexities of $\mathcal{C}(\sigma)$, $\mathcal{P}(\sigma)$ and the pre-image of $\mathcal{P}(\sigma)$, we show how to maintain item 2 and 4:

\begin{lemma}[Maintaining item~\ref{pipeline2}, similar to Theorem 7 in \cite{hoog2024fully}]\label{lemma:connectivity_maintaining_C}
   We can augment $T(\S)$ with an $O(n)$-size data structure that dynamically marks for all storing cells $C \in T(\S)$ all cells $C' \in \mathcal{C}(C)$ in a marked-ancestor tree using $O(\Psi \log \log n)$ update time.
\end{lemma}
\begin{proof}
    We augment the quadtree $T(\S)$ by keeping a counter for each node that indicates the number of storing cells in the subtree rooted at that node. Additionally, as described in item~\ref{pipeline2}, we store a marked-ancestor tree~\cite{AlstrupHR98} in which a quadtree cell $C$ is marked if there is a cell $C'$ for which $C \in \C(C')$. We again implement this using a counter that keeps track of the number of quadtree cells $C'$ for which $C \in \C(C')$. 

    We maintain our data structures as follows. When a disk $\sigma$ is inserted (or deleted), we increase (or decrease) the counter for each node on its path to the root in $O(\log \Psi)$ time. To update the marked-ancestor tree, we first find the set of constituents $\mathcal{C}(\sigma)$.
    All cells in $\C(\sigma)$ are contained in $\sigma$ and have a parent that is not contained in $\sigma$.
    That is, the set of all possible cells in the \emph{complete quadtree} that could be in $\mathcal{C}(\sigma)$ are interior disjoint and get smaller the closer they are to the boundary of $\sigma$, see Figure~\ref{fig:illustratingC}. There are $O(\Psi)$ of these cells. We can find these cells by walking through $T(\S)$, starting at the storing cell $C_\sigma$ of $\sigma$. For each of these $O(\Psi)$ cells, we simply check if its counter is greater than zero, indicating that there is a storing cell in its subtree. If this is the case, then the cell is in $\mathcal{C}(\sigma)$ and we thus mark (or unmark) it in the marked-ancestor tree in $O(\log \log n)$ time~\cite{AlstrupHR98}. It follows that we can maintain the data structure in $O(\Psi \log \log n)$ time.  
\end{proof}

\begin{lemma}[Maintaining item~\ref{pipeline4}, similar to Theorem 9 in \cite{hoog2024fully}]\label{lemma:connectivity_maintaining_P}
We can augment $T(\S)$ with a data structure that, for every storing cell $C \in T(\S)$, stores all disks in $\Gamma(C)$ in a branch-persistent disk intersection data structure.
In particular: 
\begin{itemize}
    \item For every quadtree cell $C' \in \mathcal{P}(C)$, there exists a branch storing all disks in $\Gamma(C)$ that are not in the matching $\MBM(C, C')$, and
    \item For every quadtree cell $C''$ such that $C \in \mathcal{P}(C'')$, there exists a branch storing all disks in $\Gamma(C)$ that are not in the matching $\MBM(C'', C)$.
\end{itemize}
We can maintain this augmented data structure using $O(n \log^4 n \log \Psi)$ space. Inserting/deleting a disk requires $O(\Psi \log^4 n\log \Psi)$ amortised expected time. 
\end{lemma}
\begin{proof}
The way the maximal bichromatic matching $\MBM(C,C')$ is maintained is almost identical to the maintenance of $\MBM(E,E')$ for our spanner in Section~\ref{sec:applying_branching_persistence}. The first difference is that we maintain matchings between cells instead of $\eps$-squares. Secondly, we consider the garrison $\Gamma(C)$ (equivalent to the subpopulation $\Gamma_\eps(E)$) for both cells of the matching, instead of the full population $\pi_\eps(E')$ for one of the cells. Finally, we only store the bichromatic matching between two cells $(C,C')$ if $C$ is in the perimeter of $C'$, or the other way around. I.e. $C \in \P(C')$ or $C' \in \P(C)$. Together with Lemma~\ref{lemma:bounding_size_of_C_and_P}, this implies that for every cell $C$ we store at most $|\P(C)| \in O(\Psi)$ matchings $\MBM(C,C')$ where $|C'| \leq |C|$. Furthermore, Lemma~\ref{lem:in_perimeter} implies that we consider only $O(\log \Psi)$ matchings $\MBM(C'',C)$ where$|C| \leq |C''|$ .

We store the two branch-persistent data structures as described in the lemma statement in the data structure of Lemma~\ref{lemma:intersection_queries_branching}. As before, we rebuild all of these data structures after~$K$ updates in $\S$. We set $K:= N \log \Psi / (2\Psi)$, where $N$ is $|\S|$ at the time of the last global rebuild.

To obtain the running time for rebuilding and performing the following $K$ updates, we again apply Lemma~\ref{lemma:intersection_queries_branching}. In the notation of the proof of Theorem~\ref{thm:main_small_space} we now have that $\sum_{C \in T(\S)} N_C \in O(n)$, the symmetric difference $\sum_{C \in T(\S)} z_C \in O(n\log \Psi)$ (Lemma~\ref{lemma:bounding_size_of_matching}), and the number of branches $B$ for each data structure is bounded by $O(\Psi)$ (Lemmas~\ref{lemma:bounding_size_of_C_and_P} and~\ref{lem:in_perimeter}). When inserting or deleting a disk, we only perform a single root update on each of the two data structures of the storing cell $C$ of $\sigma$. So, after $K$ updates, the total number of \emph{root-updates} is bounded by $O(K)$. For each update in $\S$, we additionally perform $O(\Psi)$ \emph{branch-updates} and queries. As in Theorem~\ref{thm:main_small_space}, it then follows by Lemma~\ref{lemma:intersection_queries_branching} that the total expected running time to rebuild and perform $K$ updates is $O((n \log \Psi + K \Psi) \log^4 n \log \Psi)$ and the expected space usage is $O((n \log \Psi + K \Psi)\log^4 n)$. Using that $K = N \log \Psi / (2\Psi)$, we obtain an amortised expected update time of $O(\Psi \log^4 n \log \Psi)$ and an expected space usage of $O(n \log^4 n \log \Psi)$.
\end{proof}

Having recovered Theorem 7 and Theorem 9 from \cite{hoog2024fully} (corresponding to Lemma \ref{lemma:connectivity_maintaining_C} and Lemma \ref{lemma:connectivity_maintaining_P}), and bounded the size of $\C(\sigma)$ and $\P(\sigma)$ by Lemma \ref{lemma:bounding_size_of_C_and_P}, we can use the identical pipeline of van der Hoog, Nusser, Rotenberg, and Staals~\cite{hoog2024fully}. 

\mainThree*

\section{Spanners and connectivity for $d$-dimensional hypercubes}\label{sec:cubes}
Throughout this section, the input $\S$ is a dynamic set of $d$-dimensional axis-aligned hypercubes with side lengths all in $[4,\Psi]$ and constant dimension $d$. We again assume that all hypercubes are contained within a bounding box $[0, \Psi^*]^d$. These assumptions can be removed by applying the technique of Appendix~\ref{app:generalisation}.

We first introduce a hypercube intersection data structure, similar to the intersection data structure for disks in Lemma~\ref{lemma:intersection_queries}. 

\hypercubelemma*

\begin{proof}
    Each $d$-dimensional hypercube induces an interval along every coordinate axis. A hypercube $\sigma \in \S$ intersects a query hypercube $\rho$ if and only if, in every coordinate direction, the corresponding intervals intersect.
We exploit this by building a multi-level structure analogous to a range tree, based on interval trees. Recall that an interval tree is a balanced decomposition of a set of intervals into elementary intervals, stored at the leaves in sorted order. Each internal node represents the union of the intervals in its subtree.

Fix an ordering of the coordinate directions (e.g., $x_1, x_2, \dots, x_d$). At the top level, we construct an interval tree on the $x_1$-intervals of the hypercubes in $\S$. For each node $\nu$ of this tree, let $\S_\nu$ denote the subset of hypercubes whose $x_1$-intervals are stored at $\nu$. At $\nu$, we recursively build a $(d-1)$-dimensional structure on the remaining coordinates, that is, an interval tree on the $x_2$-intervals of $\S_\nu$, whose nodes in turn store structures for the $x_3$-intervals, and so on. By standard range searching arguments, this hierarchical construction uses $O(n \log^{d-1} n)$ space and supports insertions and deletions in $O(\log^d n)$ time.

To answer an intersection query with a hypercube $\rho$, we proceed as follows. First, we query the top-level interval tree with the $x_1$-interval of $\rho$, which yields $O(\log n)$ nodes which together cover exactly the hypercubes whose $x_1$-intervals intersect that of $\rho$. For each such node, we recursively query its associated $(d-1)$-dimensional structure using the remaining intervals of $\rho$. At each level, this introduces an additional factor of $O(\log n)$, leading to a total query time of $O(\log^d n)$.
\end{proof}

In Theorem~\ref{thm:hypercubes_spanner}, we combine the above lemma with our techniques from Section~\ref{sec:high_space} and Appendix~\ref{sec:small_space} to obtain a dynamic $(1+\eps)$-spanner.

\hypercubesSpanner*

\begin{proof}
Recall that for our spanner results, the quadtree $T(\S)$ is defined as the minimum-size quadtree that contains the storing families $\F(\sigma)$ for all $\sigma \in \S$ (Definition~\ref{definition:quadtree}).

Before arguing about the size of the spanner, the update time, and the space usage, we make some observations with regard to $d$-dimensional hypercubes. The number of $d$-dimensional quadtree cells that a single hypercube (with side length at most~$\Psi$) can intersect is bounded by $O(\Psi^d)$. As before, each cell $C$ in $T(\S)$ is partitioned into a $d$-dimensional grid of hypercubes of side length $\eps |C|$. This means that the number of $\eps$-squares a hypercube can intersect is $O(\left( \frac{\Psi}{\eps}\right)^d)$. This is a lower bound on the update time, see Observation~\ref{observation:lower_bound}.

Since our construction is essentially the same as the one presented in Section~\ref{sec:high_space} and Appendix~\ref{sec:small_space}, we will focus on the differences. There are two main differences. The first is that the number of $\eps$-cells $E'$ that we have to consider for an $\eps$-cell $E$ has an exponential dependence on $d$, so all of our bounds will have an exponential dependence on $d$. For example, in our branch persistent data structure for hypercube connectivity, the number of branches will be $O(\left(\frac{\Psi}{\eps}\right)^d)$. The second difference is that we replace the disk intersection data structure (Lemma~\ref{lemma:intersection_queries}) with a hypercube intersection data structure (Lemma~\ref{lemma:intersection_queries_cubes}). 


\subparagraph{Update time.} As previously mentioned, the update procedure to insert or delete a hypercube~$\sigma$ is essentially identical to that of the disk graph spanner described in the proof of Theorem~\ref{thm:main_small_space}. The only difference is that for hypercube intersection, we are using the hypercube intersection data structure of Lemma~\ref{lemma:intersection_queries_cubes}. We will separately analyse the updates to the~\ref{type:one} edges corresponding to the Euclidean spanners (step 2 in Theorem~\ref{thm:main_small_space}) and \ref{type:two} edges corresponding to the pairs of $\eps$-cells (steps 3 and 4 in Theorem~\ref{thm:main_small_space}).

\begin{itemize}
    \item \textbf{Type~\ref{type:one} edges}: For a cell $C$ the area $(3 * C)$ contains $3^d$ quadtree cells in dimension~$d$. We store $O(3^d)$ Euclidean spanners for every cell $C$. As updating a single Euclidean spanner requires $O(\eps^{-d} \log n \log^2(\eps^{-1}))$ time, updating the Euclidean spanners for all cells in the storing family $\F(\sigma)$ takes $O(\eps^{-d} \log n \log \Psi \log^2(\eps^{-1}))$ time in total.
    \item \textbf{Type~\ref{type:two} edges}: For type~\ref{type:two} edges again apply Lemma~\ref{lemma:branch-persistent} to obtain a result similar to Lemma~\ref{lemma:intersection_queries_branching}. Recall that the number of branches is bounded by $O(\left(\frac{\Psi}{\eps}\right)^d)$. We obtain the following query and update times for our branch persistent version of the hypercube intersection data structure (Lemma~\ref{lemma:intersection_queries_cubes}), where $N$ is the maximum number of hypercubes stored between two rebuilds:
    \begin{itemize}
        \item Queries are supported in $O(\log^d N \log \Psi \log \eps^{-1})$ time.
        \item Let $r$ and $b$ be the number of \emph{root}- and \emph{branch-updates} since the last rebuild and $z$ the size of the symmetric difference at this rebuild. Then the rebuild and these updates take $O((\overline{n} + z + r\left(\frac{\Psi}{\eps}\right)^d) +b)\log^d N \log \Psi \log \eps^{-1} )$ total time and use $O((\overline{n} + z + r\left(\frac{\Psi}{\eps}\right)^d) +b)\log^d N)$ space, where $\overline{n}$ denotes the size of $\S$ at the rebuild.
    \end{itemize}
    
    When performing $K$ updates in $\S$, the number of \emph{root-updates} is bounded by $O(K \log \Psi)$ and the number of \emph{branch-updates} is bounded by $O(K \Psi^d \eps^{-d} \log \Psi)$. 
    
    The total time to perform a rebuild on all branch persistent data structures and perform the following $K$ updates in $\S$ is then $O((n\eps^{-d}\log \Psi + K\Psi^{d}\eps^{-d}\log\Psi)\log^d n\log\Psi\log(\eps^{-1}))$. 
\end{itemize}
By rebuilding after $K := n/(2\Psi^d)$ updates in $\S$, we obtain an amortised update time of $ O(\Psi^d\eps^{-d}\log^d n\log^2 \Psi \log^2(\eps^{-1}))$. 
\subparagraph{Size of the spanner.}
The number of type~\ref{type:one} edges is bounded by the size of the Euclidean spanners, which is $O(n\eps^{-d}\log(\eps^{-1}))$. The number of type~\ref{type:two} edges is bounded by the size of the union of all the maximal bichromatic matchings, which is $O(n\eps^{-d}\log \Psi)$. The size of the $(1+\eps)$-spanner is then $O(n\eps^{-d}\log \Psi\log(\eps^{-1}))$. 

\subparagraph{Space usage.} 
Using $K :=n/(2\Psi^d)$ in Lemma~\ref{lemma:intersection_queries_branching}, and the space bounded noted in the update time analysis for type~\ref{type:two} edgse, we find the space usage is $O((n\eps^{-d}\log\Psi+K\Psi^d\eps^{-d}\log\Psi)\log^d n)$, which simplifies to $O(n\eps^{-d}\log^d n \log \Psi)$. 
\end{proof}

In Theorem~\ref{thm:hypercubes_connectivity}, we construct a fully dynamic connectivity data structure using the techniques of Section~\ref{sec:connectivity}.

\begin{restatable}{thm}{hypercubesConnectivity}
\label{thm:hypercubes_connectivity}
Let $\S$ be a fully dynamic set of $d$-dimensional axis-aligned hypercubes in~$\mathbb{R}^d$  with diameter in $[4,\Psi]$ for a fixed and known $\Psi$ and constant dimension $d$.
We can dynamically maintain $\S$ supporting \emph{connectivity queries} in time $O(\log n / \log\log n)$ using $O(\Psi^{d-1}\log^d n \log\Psi)$ amortised update time and $O(n\log^d n \log \Psi)$ total space.
\end{restatable}
\begin{proof}
The number of squares in the perimeter of a hypercube is $O(\Psi^{d-1})$. The total number of edges in all maximal bichromatic matchings is $O(n\log\Psi)$, because for each cell $C$ with non-empty garrison $\Gamma(C)$ there are only $O(\log \Psi)$ cells $C'$ such that $|C'| = 2^{i-1}$ and $C'\subset (2^{i+5} + 1) * C$. Note that the big-O hides a constant factor $\delta^d$ dependent on the $d$-dimensional hypercube packing number.

By rebuilding after $K := (n\log \Psi)/(2\Psi^{d-1})$ updates in $\S$, following the same analysis as in Appendix~\ref{sec:connectivity}, we find the amortized update time is $O(\Psi^{d-1}\log^d n \log \Psi)$. 

Again following the same analysis as in Appendix~\ref{sec:connectivity} yields an space bound of $O((n\log \Psi + K\Psi^{d-1})\log^dn)$. For our choice of $K$,  the space usage becomes $O(n\log^d n \log\Psi )$.
\end{proof}





\end{document}